\documentclass[sigconf]{acmart}

\usepackage{multirow}
\usepackage{array}
\usepackage{graphicx}
\usepackage{arydshln}
\usepackage{bm}
\newtheorem{theorem}{Theorem}

\usepackage{enumitem}
\usepackage{braket}
\usepackage{algorithm}
\usepackage{algpseudocode}
\usepackage{soul}
\usepackage{balance}

\AtBeginDocument{%
  \providecommand\BibTeX{{%
    \normalfont B\kern-0.5em{\scshape i\kern-0.25em b}\kern-0.8em\TeX}}}

\copyrightyear{2024}
\acmYear{2024}
\setcopyright{acmlicensed}\acmConference[KDD '24]{Proceedings of the 30th ACM SIGKDD Conference on Knowledge Discovery and Data Mining}{August 25--29, 2024}{Barcelona, Spain}
\acmBooktitle{Proceedings of the 30th ACM SIGKDD Conference on Knowledge Discovery and Data Mining (KDD '24), August 25--29, 2024, Barcelona, Spain}
\acmDOI{10.1145/3637528.3671661}
\acmISBN{979-8-4007-0490-1/24/08}

\begin{document}

\title{Towards Robust Recommendation via Decision Boundary-aware Graph Contrastive Learning}

\author{Jiakai Tang}
\author{Sunhao Dai}
\author{Zexu Sun}
\affiliation{%
  \institution{\mbox{Gaoling School of Artificial Intelligence}, \mbox{Renmin University of China}}
  \city{Beijing}
  \country{China}
}
\email{tangjiakai5704@ruc.edu.cn}
\email{{sunhaodai,sunzexu21}@ruc.edu.cn}

\author{Xu Chen$^{*}$}
\author{Jun Xu}
\affiliation{%
  \institution{\mbox{Gaoling School of Artificial Intelligence}, \mbox{Renmin University of China}}
  \city{Beijing}
  \country{China}
}
\email{{xu.chen,junxu}@ruc.edu.cn}

\author{Wenhui Yu}
\author{Lantao Hu}
\author{Peng Jiang}
\author{Han Li}
\affiliation{%
  \institution{Kuaishou Technology}
  \city{Beijing}
  \country{China}
}
\email{{yuwenhui07, hulantao}@kuaishou.com}
\email{{jiangpeng, lihan08}@kuaishou.com}

\thanks{* Corresponding author}

\renewcommand{\authors}{Jiakai Tang, Sunhao Dai, Zexu Sun, Xu Chen, Jun Xu, Wenhui Yu, Lantao Hu, Peng Jiang, Han Li}
\renewcommand{\shortauthors}{Jiakai Tang et al.}


\begin{abstract}
In recent years, graph contrastive learning (GCL) has received increasing attention in recommender systems due to its effectiveness in reducing bias caused by data sparsity. However, most existing GCL models rely on heuristic approaches and usually assume entity independence when constructing contrastive views. We argue that these methods struggle to strike a balance between semantic invariance and view hardness across the dynamic training process, both of which are critical factors in graph contrastive learning.

To address the above issues, we propose a novel GCL-based recommendation framework RGCL, which effectively maintains the semantic invariance of contrastive pairs and dynamically adapts as the model capability evolves through the training process.
Specifically, RGCL first introduces decision boundary-aware adversarial perturbations to constrain the exploration space of contrastive augmented views, avoiding the decrease of task-specific information. 
Furthermore, to incorporate global user-user and item-item collaboration relationships for guiding on the generation of hard contrastive views, we propose an adversarial-contrastive learning objective to construct a relation-aware view-generator.  
Besides, considering that unsupervised GCL could potentially narrower margins between data points and the decision boundary, resulting in decreased model robustness,
we introduce the adversarial examples based on maximum perturbations to achieve margin maximization.
We also provide theoretical analyses on the effectiveness of our designs.
Through extensive experiments on five public datasets, we demonstrate the superiority of RGCL compared against twelve baseline models. To benefit the research community, we have released our project at \url{https://tangjiakai.github.io/RGCL/}.

\end{abstract}



\maketitle

\section{Introduction}\label{intro}
\begin{figure}[t]
\centering
\includegraphics[width=\linewidth]{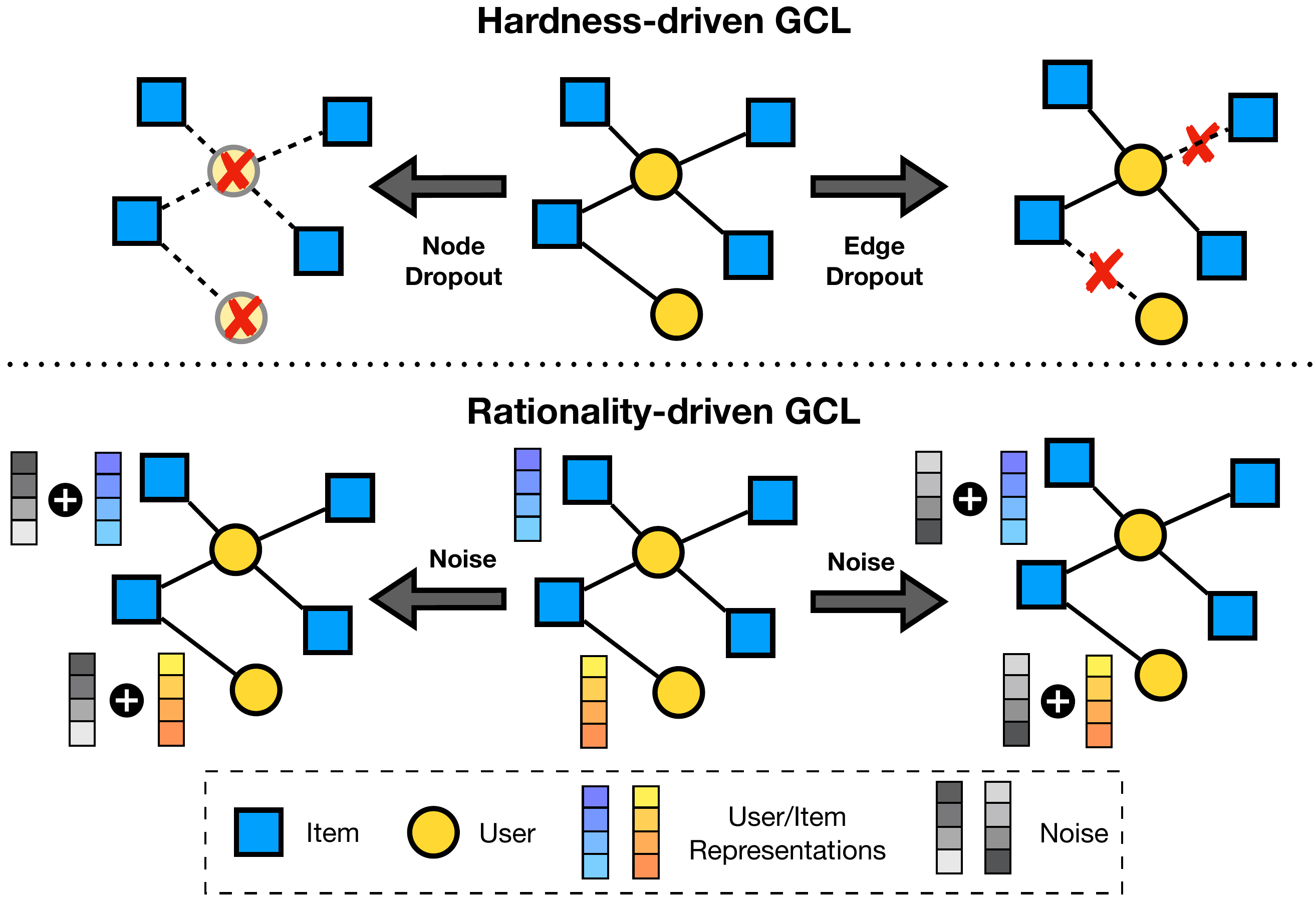} 
\caption{An overview of two types of representative GCL-based recommenders. To facilitate the presentation, we only show a single user and item with injected noise. However, in practice, the semantic-aware GCL-based methods should integrate perturbations to all graph nodes.} 
\label{fig:intro}
\end{figure}

Recently, the intersection of graph neural networks (GNNs) and recommender systems has emerged as a focal point of research attention in both academia and industry~\cite{jin2023dual}. 
While GNNs have demonstrated remarkable efficacy in capturing high-order connectivity relationships between users and items through their potent message propagation mechanism~\cite{jiao2023pglbox,wu2022graph}, the inherent data sparsity within recommendation scenarios introduces unexpected bias in users (e.g., non-active vs. active users) and items (e.g., long-tail vs. popular items) representations, thereby impairing the overall model performance~\cite{ cai2023lightgcl, lin2022improving}.

To mitigate the issue of data sparsity and drawing inspiration from self-supervised learning (SSL), recent works have introduced Graph Contrastive Learning (GCL) into GNN-based algorithms~\cite{lu2023vrkg4rec, zhang2023denoising, wei2022contrastive}. 
GCL represents a new learning paradigm that integrates contrastive learning~\cite{jaiswal2020survey} with GNN-based recommenders, simultaneously enhancing the alignment of positive embedding pairs and minimizing the similarity to augmented negative instances. In this way, GCL can effectively alleviate the problem of representation degradation among low-degree nodes. 
In general, GCL-based recommenders can be classified into two categories based on how to build the contrastive samples:
(1) \textbf{Hardness-driven methods}.
These methods basically aim to construct hard enough samples to challenge original recommender models and provide more difficult knowledge to widen the model vision.
The methods in this line mainly differentiate themselves by how to define the hardness and how to build hard enough samples. 
For example, SGL~\cite{wu2021self} generates challenging views using various strategies, such as node dropout and edge dropout.
(2) \textbf{Rationality-driven methods}.
These methods aim to maintain the rationality of the constructed samples, that is, the augmented features and original labels should form reasonable samples. For example, SimGCL~\cite{yu2022graph} makes slight changes to the original features, such that the augmented feature-label pairs can be still reasonable (\emph{i.e.}, semantically invariant).

Although the aforementioned GCL-based recommenders have shown impressive performance to some extent, we argue that these methods still suffer from several significant limitations.
As depicted in Figure~\ref{fig:intro}, \textbf{on the one hand}, hardness-driven models blindly pursue the example hardness in contrastive augmentations through manual-designed heuristic strategies. 
Unfortunately, these models may inadvertently remove certain crucial nodes or edges, neglecting how to maintain task-specific semantics.
This oversight makes it challenging for recommenders to accurately capture user preferences and item characteristics.
\textbf{On the other hand}, rationality-driven methods introduce slight feature perturbations to retain the underlying semantic structure but may overlook the benefits of introducing hard samples on providing more diverse knowledge.

Notably, both challenging positive pairs and hard negative pairs are essential to the  success of GCL-based recommenders~\cite{robinson2020contrastive, wang2021understanding}. In extreme cases, the zero-noise version of contrastive learning may not yield significant performance gains, as verified by prior research~\cite{xia2022simgrace, yu2022graph}.
In summary, achieving an adaptive and ideal balance between the hardness and rationality of contrastive augmentations for GCL-based recommenders poses a highly intricate challenge.

In this work, we aim to leverage the idea of \textit{adversarial robustness}~\cite{moosavi2016deepfool} to facilitate the construction of optimal contrastive augmented data. To be specific, the goal of adversarial robustness is to promote feature invariance upon task-relevant information, assuring the neural networks are not fooled by imperceptible data perturbations. More importantly, it specifies the maximum perturbation boundary that the current model can tolerate, which explicitly defines a feasible exploration space for conducting example augmentation. Therefore, grounded by such
idea, the graph contrastive learning can effectively balance the example hardness and rationality, both of which are crucial factors to high-quality representations.
While this idea is inherently intuitive and holds intriguing potential, its implementation still faces several challenges and obstacles.
\textbf{C1}: prevalent contrastive augmentation approaches, assuming entity independence, struggle to maintain inherent structural features as they overlook the important connections among user-user and item-item.
\textbf{C2}:  
as an unsupervised learning algorithm, GCL in blindly pursuing representation uniformity might unintentionally compromise the robust requirement, that is, narrow margins between data points and the model decision boundary, risking unexpected decreases in the model robustness.

To realize our idea and overcome the above challenges, this paper proposes a novel \textbf{R}obust \textbf{G}raph \textbf{C}ontrastive \textbf{L}earning-based recommendation framework, named RGCL.
Specifically, we first calculate the maximum perturbation magnitudes for different users and items at each graph layer, while preserving core semantic information for both user and item sides. (\textbf{Rationality}) 
Compared to manual-designed heuristics graph contrastive learning methods, we propose an adversarial-contrastive objective to adaptively generate challenging positive pairs and hard negative pairs based on the global relationships between user-user and item-item, (\textbf{Hardness}) which simultaneously overcomes the limitations of the entity independence assumption. (\textbf{C1}) 
At last, we optimize the joint loss of adversarial and contrastive components to concurrently increase the dissimilarity between different users (items) and maximize the distances between user-item inputs and model decision boundary, further improving the robustness of the recommendation model. (\textbf{C2})
In summary, our contributions can be summarized as follows:
\begin{itemize}[noitemsep, nolistsep, label=$\bullet$,leftmargin=*]
    \item We propose a model-agnostic graph contrastive learning framework, which utilizes dynamic decision boundary-aware adversarial perturbations to constrain the perturbation space of contrastive augmented view, achieving a better balance between contrastive hardness and sample rationality.
    \item We develop a joint learning algorithm based on multi-view contrastive learning and margin maximum adversarial learning to optimize RGCL, empowering better representation uniformity while improving model robustness.
    \item We give theoretical analyses to underscore the importance of hard contrastive views in model optimization and elucidate the insights behind the efficacy of RGCL in enhancing robustness.
    \item Extensive experiments on five real-world datasets demonstrate the superior performance of our proposed RGCL framework.
\end{itemize}

\section{Preliminaries}

\subsection{GNN-based Recommendation}
Formally, let $\mathcal{U} = \{u_1,u_2,\dots,u_M\}$ and $\mathcal{I}=\{i_1,i_2,\dots,i_N\}$ denote the set of users and items, respectively,
where $M$ and $N$ represent the number of users and items, respectively. Considering recommendation scenario with implicit feedback, a binary matrix $\mathbf{R}\in \mathbb{R}^{M\times N}$ are typically used to record user-item interactions (\emph{e.g.}, clicks or purchases), where $r_{u,i}=1$ indicates that user $u$ has interacted with item $i$, otherwise $r_{u,i}=0$. Following most GNN-based recommendation works~\cite{he2020lightgcn, he2023candidate, huang2021mixgcf}, we formulate the interaction behaviors between users and items as a standard bipartite graph $\mathcal{G}=\{\mathcal{V}, \mathbf{A}\}$, where $\mathcal{V} = \mathcal{U}\cup \mathcal{I}$ involves all graph nodes, and the adjacent matrix $\mathbf{A}$ is defined as follows:
\begin{equation*}
\mathbf{A}=
\left[
\begin{array}{cc}
    \mathbf{0}^{M\times M} & \mathbf{R} \\
    \mathbf{R}^T & \mathbf{0}^{N\times N}
\end{array}
\right].
\end{equation*}
Following the common practice~\cite{he2020lightgcn,cai2023lightgcl}, we encode the user $u$ and item $i$ as \textit{d}-dimensional 
latent vectors $\mathbf{e}_u\in \mathbb{R}^d$ and $\mathbf{e}_i\in \mathbb{R}^d$, respectively. 
Besides, $\mathbf{E}=\{\mathbf{e}_u\mid u \in \mathcal{U}\}\cup \{\mathbf{e}_i\mid i \in \mathcal{I}\}$ is defined as the overall learnable embedding matrix for all nodes.

Similar to other GCL-based works~\cite{yu2022graph, wu2021self, wei2022contrastive}, this paper adopts the LightGCN~\cite{he2020lightgcn} as model backbone.
Specifically, the comprehensive graph representations $\mathbf{z}_u$ and $\mathbf{z}_i$ for user $u$ and item $i$ in LightGCN are calculated by 
\begin{equation*}
    \small    
    \setlength{\abovedisplayskip}{2pt}
    \setlength{\belowdisplayskip}{2pt}
    \setlength{\abovedisplayshortskip}{2pt}
    \setlength{\belowdisplayshortskip}{2pt}
    \begin{aligned}
        \mathbf{z}_{u} &= \sum_{l=0}^{L} \mathbf{h}_{u}^{(l)},\quad
        \mathbf{h}_{u}^{(l)} = \sum_{j\in \mathcal{N}_u}\frac{1}{\sqrt{|\mathcal{N}_u| |\mathcal{N}_j|}} \mathbf{h}_{j}^{(l-1)},\quad l\geq1,\\
        \mathbf{z}_{i} &= \sum_{l=0}^{L} \mathbf{h}_{i}^{(l)},\quad
        \mathbf{h}_{i}^{(l)} = \sum_{v\in \mathcal{N}_i}\frac{1}{\sqrt{|\mathcal{N}_i| |\mathcal{N}_v|}} \mathbf{h}_{v}^{(l-1)},\quad l\geq 1,
    \end{aligned}
\end{equation*}
where $\mathcal{N}_u$ and $\mathcal{N}_i$ indicate the neighboring nodes of user $u$ and item $i$, respectively. $\mathbf{h}^{(l)}_u$ and $\mathbf{h}^{(l)}_i$ means the $l$-th layer graph representation for user $u$ and item $i$, respectively. Here, $\mathbf{h}^{(0)}_u$ and $\mathbf{h}^{(0)}_i$ are initialized with the learnable embedding $\mathbf{e}_u$ and $\mathbf{e}_i$, respectively.
The predicted score $\hat{r}_{u,i}$ for the $(u,i)$ pair is computed as the inner product of their graph representations, \emph{i.e.}, $\hat{r}_{u,i}=\Braket{\mathbf{z}_u, \mathbf{z}_i}$.
Finally, the BPR~\cite{rendle2012bpr} loss is adopted as the optimization objective:
\begin{equation}
    \label{bpr}
    \begin{aligned}
        \mathcal{L}_{BPR} = -\sum_{u\in \mathcal{U}}\sum_{i^+\in\mathcal{I}_u^{+}}\sum_{i^-\in \mathcal{I}_{u}^{-}}\ln \sigma(\hat{r}_{u,i^+} - \hat{r}_{u,i^-}),
    \end{aligned}
\end{equation}
where $\sigma(x)=1/(1+e^{-x})$, $\mathcal{I}_u^{+}$ and $\mathcal{I}_u^{-}$ represent the positive item and unobserved item set for user $u$, respectively.

\subsection{GCL-based Recommenders}
In real-world scenarios, interaction behaviors between users and items are actually highly sparse, which can lead to severe overfitting and bias problems~\cite{wu2021self, jing2023contrastive}.
Graph contrastive learning (GCL), as a novel learning paradigm, helps mitigate the above problems~\cite{yu2022graph,cai2023lightgcl}. 
In specific, GCL firstly generates diverse graph views for each user and item (\emph{e.g.}, node dropout and feature masking). Then the different views of the same user (item) are treated as the positive pairs, while the different views of the different instances are treated as the negative pairs. Finally, contrastive learning loss is used to optimize the model parameters with paired users and items, where InfoNCE~\cite{oord2018representation} is the most commonly adopted loss. Formally, the contrastive learning loss for the user side can be defined as follows:
\begin{equation}
    \label{origin_cl_u}
    \begin{aligned}
        \mathcal{L}_{CL}^{U}(\mathbf{x}_{u},\mathbf{y}_{u}) = \sum_{u\in\mathcal{U}} -\log  \frac{\exp(sim(\mathbf{x}_u, \mathbf{y}_u)/\tau)}{\sum_{v\in\mathcal{U}}\exp(sim(\mathbf{x}_u, \mathbf{y}_v)/\tau)},
    \end{aligned}
\end{equation}
where $\mathbf{x}_u$ and $\mathbf{y}_u$ denote the two different augmented views of user $u$, $sim(\cdot,\cdot)$ and $\tau$ represents the cosine similarity function and temperature hyper-parameter, respectively. Similarly, the contrastive learning loss of the item side is formulated as follows:
\begin{equation}
    \label{origin_cl_i}
    \begin{aligned}
        \mathcal{L}_{CL}^{I}(\mathbf{x}_{i},\mathbf{y}_{i}) = \sum_{i\in\mathcal{I}} -\log  \frac{\exp(sim(\mathbf{x}_i, \mathbf{y}_i)/\tau)}{\sum_{j\in\mathcal{I}}\exp(sim(\mathbf{x}_i, \mathbf{y}_j)/\tau)}.
    \end{aligned}
\end{equation}
where $\mathbf{x}_i$ and $\mathbf{y}_i$ denote the two different views of item $i$.

\subsection{Adversarial Robustness}
Adversarial training (AT) stands out as one of the most promising approaches for bolstering adversarial robustness~\cite{moosavi2016deepfool, madry2017towards, goodfellow2014explaining}. The goal of AT is to increase model robustness by generating adversarial examples through well-designed perturbations, which purposefully induce the neural network to error. 
Formally, the optimal perturbation for data sample $(x,y)$ is found by maximizing the loss function $\mathcal{L}(\cdot):\delta^* = \arg \max \mathcal{L}(x+\delta,y;\bm{\theta})$ where
$\delta$ represents an adversarial perturbation of $\ell_p$ norm smaller than $\epsilon$. Then, the model is trained on a mixture of both original clean examples and generated adversarial examples to enhance the robustness ability.

\textbf{Discussion.} 
Adversarial robustness uncovers the root cause of the model's adversarial vulnerability, that is, the non-smooth feature space near data samples~\cite{jiang2020robust}. In other words, small input perturbations likely result in large changes in the potential semantics, subsequently affecting the model output, which is the basis challenge that adversarial defense algorithms strive to resolve. Actually, this particularly fits well with graph contrastive learning, which aims to maximize the consistency of the given instance under different augmentation views. More importantly, adversarial robustness provides the maximum boundary of feature perturbations that the model can tolerate (\emph{cf.} Sec~\ref{dicision_boundary_aware_pert}), which effectively restrains the exploration space for contrastive augmentation and guides the generation of optimal view-generator.

\section{Our Approach: RGCL}

\begin{figure*}[t]
\centering
\includegraphics[width=\linewidth]{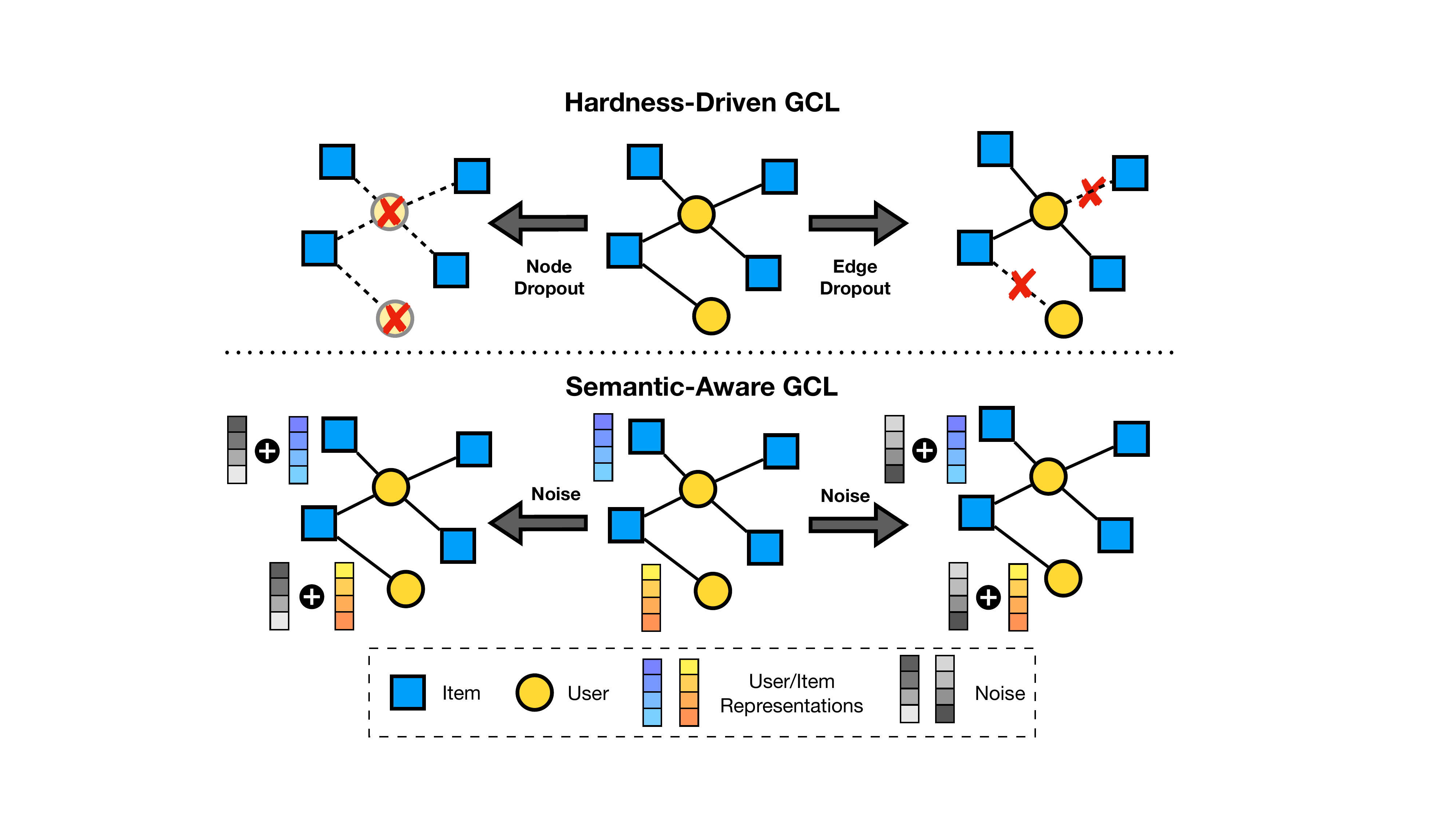} 
\caption{Overall framework of our proposed dynamic decision boundary-aware graph contrastive learning framework RGCL.}
\label{fig:framework}
\end{figure*}

\subsection{Overall Framework}
The overall framework of RGCL is presented in Figure~\ref{fig:framework}. 
In specific, we calculate the maximum feature perturbations to guide the subsequent generation of both contrastive examples and adversarial examples. For contrastive examples, we firstly generate two random-augmented views $\mathbf{Z}^{\prime}$ and $\mathbf{Z}^{\prime\prime}$ using random perturbations. 
Besides, the third view $\mathbf{Z}^{ac}$, which we refer to as adversarial-contrastive view, is generated through maximizing relation-aware contrastive function. On the foundation of these contrastive samples, we employ multi-view contrastive learning to prompt high-quality representations.
Furthermore, to safeguard the model robustness against potential compromises arising from the uniformity optimization of graph contrastive learning, we generate adversarial examples using maximum perturbation to strenuously enlarge the distances between data points and the decision boundary. Finally, the model is updated by employing a joint optimization objective with augmented contrastive and adversarial data.

\subsection{Decision Boundary-aware Perturbation}\label{dicision_boundary_aware_pert}
To build our contrastive samples, we first derive perturbations that the original samples can maximally tolerate to maintain user preferences.
Ideally, the perturbations should satisfy two conditions:
(1) the perturbations should be as large as possible, such that the obtained contrastive samples are hard enough (\textbf{hardness requirement}).
(2) The augmented samples after incorporating the perturbations should be still aligned with the user's original preferences (\textbf{rationality requirement}).

Different from traditional adversarial learning problems based on classification settings, recommender system is basically a ranking problem, and the perturbations should be learned to maintain user preference rankings. To this end, we propose to learn the maximum perturbations that can maintain item pair-wise rankings. Furthermore, given that different orders of graph representations possess different levels of expressive capacity, that is, higher-layer representations aggregate richer structure information and reflect more complex connectivity patterns. Consequently, we tailor the maximum perturbation for each high-order graph representation independently. In specific, for each user $u$ and a positive-negative item pair $(i^+,i^-)$, suppose their original representations are 
$\mathbf{{z}}_{u} = \sum_{l=0}^{L}\mathbf{h}_{u}^{(l)}$, $\mathbf{z}_{i^+}^{\top} = \sum_{l=0}^{L}\mathbf{h}_{i^+}^{(l)}$, and $\mathbf{z}_{i^-}^{\top} = \sum_{l=0}^{L}\mathbf{h}_{i^-}^{(l)}$, respectively. We define the pair-wise ranking function as $g(u,i^+,i^-)=\Braket{\mathbf{\tilde{z}}_{u}^{(k)},\mathbf{{z}}_{i^+}}-\Braket{\mathbf{\tilde{z}}_{u}^{(k)},\mathbf{{z}}_{i^-}}$, where $\mathbf{\tilde{z}}_{u}^{(k)} = \sum_{l=0,l\neq k}^{L}\mathbf{h}_{u}^{(l)} + (\mathbf{h}_{u}^{(k)} + \bm{\Delta})$ is the user embedding after incorporating perturbation $\bm{\Delta}\in\mathbb{R}^d$ to $k$-th layer graph representation $\mathbf{h}_u^{(k)}$, and $<\cdot,\cdot>$ means inner product. Then, the learning objective of perturbation $\bm{\Delta}$ is designed as follows:
\begin{equation}
\mathbf{\Delta}_{u}^{(k)} = \arg\max_{\bm{\Delta}} ||\bm{\Delta}||_p \quad \text{s.t. }  g(u,i^+,i^-)>0,
\end{equation}
where $\Vert\cdot\Vert_p$ means the vector's p-norm. Here, pair-wise ranking function $g(\cdot)$ is linearized around the $k$-th representation $\mathbf{h}_{u}^{(k)}$, thus the maximum perturbation $\bm{\Delta}_u^{(k)}$ is exactly corresponding to the orthogonal projection of $\mathbf{h}_{u}^{(k)}$ onto the model decision hyperplane. 

For the sake of simplicity and better interpretation, we denote that 
$f(\mathbf{h}_{u}^{(k)}) = \partial{g(u,i^+,i^-)}/\partial{\mathbf{h}_{u}^{(k)}}$. The maximum perturbation $\mathbf{\Delta}_u^{(k)}$ is equivalent to solving for the directional vector from $\mathbf{h}_u^{(k)}$ to the decision boundary, which is formally given as follows:
\begin{equation}
    \mathbf{\Delta}_{u}^{(k)} = -\frac{g(u,i^+,i^-)}{\Vert f(\mathbf{h}_{u}^{(k)})\Vert_{q}^{q}}\cdot \text{sign}(f(\mathbf{h}_{u}^{(k)})) \odot \Vert f(\mathbf{h}_{u}^{(k)})\Vert^{q-1},
\end{equation}
where $\text{sign}(\cdot)$ is the sign function, and $\odot$ denotes element-wise product. The value of $q$ depends on the choice of perturbation norm $\ell_p$ ($1\leq p \leq \infty$), and satisfies that $\frac{1}{p}+\frac{1}{q}=1$ by following Holder's Inequality's constraint~\cite{moosavi2016deepfool}. In our work, $p$ is set as $\infty$ and $q$ is set as 1, as we empirically found that perturbation constraints under the $\ell_{\infty}$ norm have better model performance.

Following that, since users often interact with multiple items in real-world recommendation scenarios, we extend the above method to all interactions of user $u$ for deriving the final optimal perturbation constraint, which can be rewritten as follows:
\begin{equation}
    \label{max_pert_vec}
    \begin{aligned}
        \mathbf{\Delta}_{u}^{(k)} = -\frac{g(u,i^+,i^-)}{\Vert f(\mathbf{h}_{u}^{(k)})\Vert_{1}}\cdot \text{sign}(f(\mathbf{h}_{u}^{(k)})),\\
        \text{where}\ i^{+},i^{-} = \underset{i^+\in\mathcal{I}_u^{+}, i^{-}\in\mathcal{I}_u^{-}}{\arg\min} \left| \frac{g(u,i^+,i^-)}{\Vert f(\mathbf{h}_{u}^{(k)})\Vert_{1}} \right|.
    \end{aligned}
\end{equation}

Note that we only focus on perturbing the high-order graph representations for users and items, while skipping the beginning features, \emph{i.e.}, $1\leq k \leq L$. This is because the original features contain the most abundant semantic information, and polluting these features could lead to a severe performance decrease. On the other hand, by perturbing higher-order representations, we subtly and implicitly disrupt the potential semantic and structural characteristics. Intuitively, it can efficaciously simulates the noise encountered in real-world scenarios, thereby further enhancing the model robustness.
Similarly, we can obtain the graph perturbations of item nodes from a dual perspective.

\subsection{Relation-aware Contrastive Learning with Perturbation Constraints}
As highlighted in Sec.~\ref{intro}, existing GCL-based recommenders struggle to achieve a harmonious balance between contrastive hardness and rationality, both of which are pivotal to acquire high-quality user (item) representations. 
To this end, in this subsection, we meticulously design the relation-aware adversarial-contrastive objective, which utilizes the global relationships among user-user and item-item to create more challenging positive and hard negative pairs under perturbation constraints. Finally, we optimize the representations through multi-view contrastive learning.

\subsubsection{\textbf{Perturbation-constrained Contrastive Augmentation}}

Following previous works~\cite{yu2022graph,yu2023xsimgcl}, we adopt the random perturbations 
$\{\mathbf{r}_u^{(l)}:l=1,2,\cdots,L\}$ for user $u$ to generate the first random contrastive view $\mathbf{z}_u^{\prime}$ as follows:
\begin{equation}
    \label{random_aug}
    \begin{aligned}
        \mathbf{z}_{u}^{\prime} = \frac{1}{L+1}\left(\mathbf{h}_{u}^{(0)} + \sum_{l=1}^L\left(\mathbf{h}_{u}^{(l)} + \mathbf{r}_{u}^{(l)}\right)\right), \\ 
        \text{where\quad} \mathbf{r}_{u}^{(l)} = \epsilon\cdot \frac{\mathbf{r} \odot \text{sign}(\mathbf{h}_{u}^{(l)})}{\Vert \mathbf{r} \odot \text{sign}(\mathbf{h}_{u}^{(l)}) \Vert_2}.
    \end{aligned}
\end{equation}
Here, $\mathbf{r} \in\mathbb{R}^{d}$ following a uniform distribution $U(0,1)$, and $\epsilon$ is a hyper-parameter to control the initial perturbation magnitude. Similarly, we could obtain the augmentation views $\mathbf{z}_i^{\prime}$ for item $i$.

Following that, we can get the second augmented representations $\mathbf{z}_{u}^{\prime\prime}$ and $\mathbf{z}_{i}^{\prime\prime}$ in the same way but utilizing the perturbations $\textbf{r}$ with different random initialization for more diverse contrastive effects.

However, different users and items have unique \textit{intrinsic robustness}, which means that even imperceptible perturbations may result in large semantic changes for fragile instances. In turn, they unintentionally lead to the erroneous feature-label examples, which is heavily overlooked by existing GCL methods. Therefore, we propose to employ the instance-wise perturbation constrains to guide the generation of contrastive samples, aiming to avoid lossing task-relevant semantic information and build rational view-generator. 
Specifically, for the $l$-layer augmentation perturbations $\mathbf{r}_u^{(l)}$, we constrain its exploration space by using the following projection operation $\Pi(\cdot)$ to obtain the constrained perturbation $\tilde{\mathbf{r}}_u^{(l)}$:
\begin{equation}
    \label{projection}
    \mathbf{\tilde{r}}_u^{(l)} = \Pi(\mathbf{r}_u^{(l)})= \min (abs(\bm{\Delta}_u^{(l)}), \max (-abs(\bm{\Delta}_u^{(l)}), \mathbf{r}_u^{(l)}),
\end{equation}
where $\max(\cdot,\cdot)$ and $\min(\cdot,\cdot)$ are both wise-element operations, and $abs(\cdot)$ computes the absolute value of each element for the given vector.
Here, we conservatively constrain the magnitude of random perturbation $\mathbf{\tilde{r}}_u^{(l)}$ within a bounded $\delta_u^{(l)}$-ball, where we define $\delta_u^{(l)}$ as $||\mathbf{\Delta}_u^{(l)}||_{\infty}$. The main motivation behind Eq.~(\ref{projection}) is that $\mathbf{\Delta}_u^{(l)}$ is the maximum perturbation with the most attacking direction, and our conservative strategy ensures that other perturbation direction bounded within the ball could also safely maintain semantic invariance. Consequently, we replace $\mathbf{r}_u^{(l)}$ in Eq.~(\ref{random_aug}) with constrained perturbation $\mathbf{\tilde{r}}_u^{(l)}$ for achieving contrastive rationality.

\subsubsection{\textbf{Relation-aware Adversarial-Contrastive Augmentation}}
To break the assumption of instance independence in traditional GCL-based algorithms and simultaneously further enhance the hardness of contrastive examples, RGCL generates the relation-aware adversarial-contrastive perturbations to fool the model by confusing the identities among different users and items. 
To be specific, we propose to maximize the following contrastive loss for generating instance-specific perturbations $\bm{\eta}$:
\begin{gather}
\label{eta_max}
        \underset{\bm{\eta}}{\max} \sum_{u\in\mathcal{U}} -\log  \frac{\exp(sim(\ddot{\mathbf{z}}_u, \mathbf{z}_u^{\prime\prime})/\tau)}{\exp(sim(\ddot{\mathbf{z}}_u, \mathbf{z}_u^{\prime\prime})/\tau)+\sum_{v\in\mathcal{U}/{u}}\exp(sim(\ddot{\mathbf{z}}_u, \mathbf{z}_v^{\prime\prime})/\tau)}, \notag \\
        \text{where}\ \ddot{\mathbf{z}}_u = \frac{1}{L+1}\left(\mathbf{h}_{u}^{(0)} + \sum_{l=1}^L\left(\mathbf{h}_{u}^{(l)} + \mathbf{\tilde{r}}_u^{(l)} + \bm{\eta}_{u}^{(l)}\right)\right),
\end{gather}
and $\bm{\eta} = \{||\bm{\eta}_{u}^{(l)}||_\infty\leq \delta_u^{(l)}: u \in \mathcal{U}, l \in \{1,2,\dots,L\}\}$ denotes the perturbation set of user $u$. However, as the general GNN-based recommenders involve nonlinear transformations, it is extremely challenging to find a closed-form solution for the above optimization problem. Drawing inspiration from the fast gradient sign method (FGSM) proposed in Goodfellow \emph{et al.}~\cite{goodfellow2014explaining}, which assumes that the objective function is approximately linear around the current model parameters. Building on this approximation, we can obtain an optimal max-norm constrained perturbation as follows:
\begin{equation}
\label{eta_fgsm}    
    \bm{\eta}_u^{(l)} = \delta_u^{(l)}\cdot \text{sign}(\partial \mathcal{L}_{CL}^U(\ddot{\mathbf{z}}_u,\mathbf{z}_u^{\prime\prime})/\partial \bm{\eta}_u^{(l)}).
\end{equation}

Similarly, we can compute the relation-aware perturbations for items. Due to space limitation, the detailed derivation steps are omitted here.
After that, we generate the relation-aware adversarial-contrastive views for users and items as follows:
\begin{equation}
\label{ac_view}
    \small    
    \begin{aligned}
        \mathbf{z}_{u}^{ac} &= \frac{1}{L+1}\left(\mathbf{h}_{u}^{(0)} + \sum_{l=1}^L\left(\mathbf{h}_{u}^{(l)} + \tilde{\mathbf{r}}_u^{(l)} \odot \text{sign}(\bm{\eta}_u^{(l)})\right)\right), \\
        \mathbf{z}_{i}^{ac} &= \frac{1}{L+1}\left(\mathbf{h}_{i}^{(0)} + \sum_{l=1}^L\left(\mathbf{h}_{i}^{(l)} + \tilde{\mathbf{r}}_i^{(l)} \odot \text{sign}(\bm{\eta}_i^{(l)})\right)\right),
    \end{aligned}
\end{equation}
where $\tilde{\mathbf{r}}_u^{(l)}$ and $\tilde{\mathbf{r}}_i^{(l)}$ are defined in Eq.~(\ref{projection}) and note that they are initialized with different random values.

Compared to the random-augmented view, adversarial-contrastive augmentation has two main advantages: (1) The optimization objective integrates global users (items) to confuse their identities, thus the view generation process is essentially guided by the user-user and item-item relationships, resulting in relation-aware and more challenging contrastive representations. (2) Considering different intrinsic vulnerability among instances, our proposed adversarial-contrastive perturbations are instance-specific and dynamically adopted along with the model training process, thereby further improving the model robustness and adaptability.

\subsubsection{\textbf{Multi-View Contrastive Learning}}
In summary, based on the above discussion, we have obtained views triplets $(\mathbf{z}_u^{\prime}, \mathbf{z}_u^{\prime\prime}, \mathbf{z}_u^{ac})$ and $(\mathbf{z}_i^{\prime}, \mathbf{z}_i^{\prime\prime}, \mathbf{z}_i^{ac})$ for user $u$ and item $i$, respectively. Then, we employ multi-view contrastive learning objective for different views of the same instances, \emph{i.e.}, $\{\mathbf{z}_{u}^{\prime}\leftrightarrow\mathbf{z}_{u}^{\prime\prime},\mathbf{z}_u^{ac}\leftrightarrow\mathbf{z}_{u}^{\prime}, \text{ and } \mathbf{z}_u^{ac}\leftrightarrow\mathbf{z}_{u}^{\prime\prime}\}$ for user $u$, while $\mathbf{z}_{i}^{\prime}\leftrightarrow\mathbf{z}_{i}^{\prime\prime}, \mathbf{z}_{i}^{ac}\leftrightarrow\mathbf{z}_{i}^{\prime}, \text{ and } \mathbf{z}_{i}^{ac}\leftrightarrow\mathbf{z}_{i}^{\prime\prime}$ for item $i$.

The complete contrastive loss function is formulated as follows: 
\begin{equation}        
\label{multi_view}
    \begin{aligned}   
        \mathcal{L}_{CL} 
        =& \mathcal{L}_{CL}^{U}(\mathbf{z}_u^{\prime},\mathbf{z}_{u}^{\prime\prime})+ \mathcal{L}_{CL}^{U}(\mathbf{z}_u^{ac},\mathbf{z}_{u}^{\prime}) + \mathcal{L}_{CL}^{U}(\mathbf{z}_u^{ac},\mathbf{z}_{u}^{\prime\prime}) \\
        & \mathcal{L}_{CL}^{I}(\mathbf{z}_i^{\prime},\mathbf{z}_{i}^{\prime\prime})+
        + \mathcal{L}_{CL}^{I}(\mathbf{z}_i^{ac},\mathbf{z}_{i}^{\prime}) + \mathcal{L}_{CL}^{I}(\mathbf{z}_i^{ac},\mathbf{z}_{i}^{\prime\prime}).
    \end{aligned}                                        
\end{equation}
where $\mathcal{L}_{CL}^{U}(\cdot)$ and $\mathcal{L}_{CL}^{I}(\cdot)$ are defined in Eq.~(\ref{origin_cl_u}) and (\ref{origin_cl_i}), respectively.
Through the multi-view contrastive learning approach, the model is able to acquire more difficult knowledge from hard yet rational contrastive pairs, mitigating recommendation biases and preventing the overfitting resulting from sparse supervised data.

\subsection{Towards Margin Maximization via Adversarial Optimization}
However, excessive pursuit of representation uniformity in GCL may lead to reduced distances between data points and the decision boundary, potentially compromising the model robustness. We attribute such dilemma is caused by the inherent deficiency that the GCL's essence is unsupervised learning paradigm, which pushes all different instances apart while ignoring task-specific semantic relations~\cite{wang2021understanding}.
To tackle the above issue, we propose to use adversarial examples for achieving margin maximization. Specifically, we generate adversarial examples using the maximum adverasrial perturbation defined in Eq.~(\ref{max_pert_vec}), which can be formulated as follows:
\begin{equation}
\setlength{\abovedisplayskip}{3pt}
\setlength{\belowdisplayskip}{3pt}
\setlength{\abovedisplayshortskip}{3pt}
\setlength{\belowdisplayshortskip}{3pt}
\label{adv_example}
    \begin{aligned}
        \mathbf{z}_{u}^{adv} &= \frac{1}{L+1}\left(\mathbf{h}_{u}^{(0)} + \sum_{l=1}^L\left(\mathbf{h}_{u}^{(l)} + \mathbf{\Delta}_{u}^{(l)}\right)\right), \\
        \mathbf{z}_{i}^{adv} &= \frac{1}{L+1}\left(\mathbf{h}_{i}^{(0)} + \sum_{l=1}^L\left(\mathbf{h}_{i}^{(l)} + \mathbf{\Delta}_{i}^{(l)}\right)\right).
    \end{aligned}
\end{equation}

We then utilize the generated adversarial examples to optimize the BPR objective (\emph{i.e.}, Eq.~(\ref{bpr})), which is given as follows:
\begin{equation}
\label{adv_loss}    
    \begin{aligned}
        \mathcal{L}_{ADV} = -\sum_{u\in \mathcal{U}}\sum_{i^+\in\mathcal{I}_u^{+}}\sum_{i^-\in \mathcal{I}_{u}^{-}}\ln \sigma(\hat{r}_{u,i}^{adv} - \hat{r}_{u,j}^{adv}), \\
        \text{where}\ 
        \hat{r}_{u,i}^{adv}=\Braket{\mathbf{z}_{u}^{adv},\mathbf{z}_{i}^{adv}},\ 
        \hat{r}_{u,j}^{adv}=\Braket{\mathbf{z}_{u}^{adv},\mathbf{z}_{j}^{adv}}.
    \end{aligned}
\end{equation}

By explicitly creating adversarial examples around the model's decision boundary, the model optimized with both original and adversarial data can more effectively boost the confidence of input data, thereby enhancing the model's overall robustness.

\subsection{Model Training}
\subsubsection{\textbf{Joint Optimization Objective}}
In the training stage, we propose to optimize the model parameters with the joint learning objective, which is formulated as follows:
\begin{equation}\label{final_loss}
    \mathcal{L} = \mathcal{L}_{BPR} + \mu \mathcal{L}_{ADV} + \alpha\mathcal{L}_{CL},
\end{equation}
where $\mu$ and $\alpha$ are the hyper-parameters for different loss terms. 

\subsubsection{\textbf{Complexity Analysis}}
Since RGCL doesn't introduce any other trainable parameters, the space complexity and the inference time complexity of model remains the same as GNN backbone. 
Besides, the total training time complexity of RGCL is $O((L|\mathcal{E}|+B^2)d)$, where $B$ and $\mathcal{E}$ denote the batch size and edge set, respectively. Thus, our method retains the same order of computation complexity as other state-of-the-art GCL-based methods, such as SimGCL~\cite{yu2022graph} and RocSE~\cite{ye2023towards}. Due to the limited space, please refer to Appendix~\ref{time_complexity} for more detailed analysis.
\section{Theoretical Analysis}
\subsection{Hardness-aware Contrastive Learning}\label{theory_1}
The core motivation of this paper is to construct \textbf{semantic preserving} and \textbf{hardness enhancing} view-generator for contrastive learning. 
For the former, 
we capitalize on the decision boundary-aware constraint to help build rationality-aware views. 
For the latter, we carefully construct more challenging contrastive paired data because their hardness significantly affects the optimization process of model parameters. 

To further explain, we give a proof that contrastive loss is essentially hardness-aware learning mechanism.
Specifically, taking the side of users as an example, given a set of users $\mathcal{U}=\{u_1,u_2,\dots,u_M\}$, we denote the similarity of user $u_{i}$ under different augmented views (\emph{e.g.}, random-augmented view or adversarial-contrastive view) as $s_{i,i}$, and the similarity between user $u_{i}$ and $u_{j}$ as $s_{i,j}$. The probability of $u_i$ being identified as $u_j$ is formulated as:
\begin{equation*}
    P_{i,j} = \frac{\exp(s_{i,j}/\tau)}{\exp(s_{i,i}/\tau) + \sum_{k \neq i} \exp(s_{i,k}/\tau))}.
\end{equation*}
Thus, the objective of contrastive learning is rewritten as follows:
\begin{equation*}
    \varphi(u_i) = -\log\frac{\exp(s_{i,i}/\tau)}{\exp(s_{i,i}/\tau) + \sum_{k \neq i} \exp(s_{i,k}/\tau)}.
\end{equation*}
Then, the expression of updating model parameters $\bm{\theta}$ is
\begin{equation*}
    \frac{\partial \varphi(u_i)}{\partial\bm{\theta}} = \frac{\partial \varphi(u_i)}{\partial{s_{i,i}}} \frac{\partial{s_{i,i}}}{\partial\bm{\theta}} + \sum_{j \neq i}\frac{\partial \varphi(u_i)}{\partial{s_{i,j}}} \frac{\partial{s_{i,j}}}{\partial\bm{\theta}},
\end{equation*}
where we give the derivation results for $\frac{\partial \varphi(u_i)}{\partial{s_{i,i}}}$ and $\frac{\partial \varphi(u_i)}{\partial{s_{i,j}}}$:
\begin{equation}
    \begin{aligned}
        \frac{\partial \varphi(x_i)}{\partial s_{i,i}} &= \frac{1}{\tau} (P_{i,i}-1) \propto \exp(s_{i,i}/\tau),\\        \frac{\partial \varphi(u_i)}{\partial{s_{i,j}}} &= \frac{1}{\tau} P_{i,j} \propto \exp(s_{i,j}/\tau),
    \end{aligned}
\end{equation}
where we can observe that the gradients of the contrastive loss \emph{w.r.t.} both positive and negative pairs are proportional to the corresponding exponential form of their similarity scores. 
This means that smaller positive pair similarity $s_{i,i}$ and larger negative pair similarity $s_{i,j}$ will have a greater impact on the model parameter optimization.
Therefore, our proposed RGCL can learn the high-quality representations by constructing the challenging positive pairs and hard negative pairs, which fits to guide model optimization through hardness-aware contrastive learning.

\subsection{Theoretical Analysis of Model Robustness}
Although contrastive learning can improve the representation uniformity and reduce the recommendation bias, it may potentially push data points closer to model decision boundary and eventually decrease model robustness due to the nature of task-unrelated unsupervised learning. To make it up, our RGCL explicitly maximizes the margin by constructing adversarial examples based on decision boundary-aware perturbation. Then, in this subsection, we give the explanation on the rationality of our method. 

For the sake of notation simplicity, we assume that input example is denoted as $x$. The goal of recommendation algorithm is to make the preference probabilities for user $u$'s positive items are higher than that for negative items, which is denoted as $g(x;\bm{\theta})>0$. Inspired by work~\cite{ding2018mma}, the margin between data point and decision boundary is denoted as $d(x;\bm{\theta})$, which can be defined as follows:
\begin{equation}
    d(x;\bm{\theta})=\Vert \mathbf{\Delta}^* \Vert = \max \Vert \mathbf{\Delta} \Vert\quad s.t.\ \mathbf{\Delta}: g(x+\mathbf{\Delta};\bm{\theta})> 0.
\end{equation}
We denote the BPR loss function as $\psi(\cdot)$, then we have the theorem:
\begin{theorem}
    Gradient descent on $\psi(g(x+\mathbf{\Delta}^*;\bm{\theta}))$ \emph{w.r.t.} $\bm{\theta}$ with a proper step size increases $d(x;\bm{\theta})$, where $\mathbf{\Delta}^*=\arg \max_{g(x+\mathbf{\Delta};\bm{\theta})> 0} \Vert \mathbf{\Delta} \Vert$ is the maximum perturbation given the current $\bm{\theta}$.
\end{theorem}
\begin{proof}
    Let $\rho(\mathbf{\Delta})=\Vert \mathbf{\Delta} \Vert$ and assume $\rho(\mathbf{\Delta})$ and $\psi(g(x;\bm{\theta}))$ are functions with twice continuous derivatives in a neighborhood of $(\mathbf{\Delta}^*,\bm{\theta})$, $c$ is a constant, and the matrix 
    $$\begin{pmatrix}
        \frac{\partial^2 \rho(\mathbf{\Delta}^*)}{\partial \mathbf{\Delta}^2} + c \cdot \frac{\partial^2 \psi(g(x+\mathbf{\Delta}^*;\bm{\theta}))}{\partial \mathbf{\Delta}^2} & \frac{\psi(g(x+\mathbf{\Delta}^*;\bm{\theta}))}{\partial \mathbf{\Delta}} \\
        \left(\frac{\partial \psi(g(x+\mathbf{\Delta}^*;\bm{\theta}))}{\partial \mathbf{\Delta}}\right)^T & 0
    \end{pmatrix}$$
    is full rank, then we have 
    \begin{equation*}
        \nabla d(x;\bm{\theta}) = C(x,\bm{\theta}) \frac{\partial \psi(g(x+\mathbf{\Delta}^*;\bm{\theta}))}{\partial \theta},
    \end{equation*}
\noindent
where
\begin{equation*}
    C(x,\bm{\theta}) = \frac{\left\langle \frac{\partial \rho(\mathbf{\Delta}^*)}{\partial \mathbf{\Delta}}, \frac{\partial \psi(g(x+\mathbf{\Delta}^*;\bm{\theta}))}{\partial \mathbf{\Delta}} \right\rangle}{\left\| \frac{\partial \psi(g(x+\mathbf{\Delta}^*;\bm{\theta}))}{\partial \mathbf{\Delta}} \right\|_2^2}
\end{equation*}
is a scalar. 
\end{proof}
The above proof demonstrates that under proper perturbations, our method can maximize the margin by minimizing the adversarial loss. 
Therefore, our proposed method can maximize the margin between data points and the model decision boundary by generating adversarial examples with the maximum perturbations defined in Seq.~\ref{dicision_boundary_aware_pert}, thereby effectively improving the robustness of model.
Besides, we give an additional robust analysis of our method from the perspective of connections between the sharpness of loss landscape and PAC-Bayes theory. It further theoretically elaborates on the model's tolerance to parameter perturbations. The detailed analysis is presented in the Appendix~\ref{theory}.

\begin{table*}[h]
    \renewcommand{\arraystretch}{1.1}
    \centering
    \caption{Overall performance comparison among baseline and our models. We use bold fonts to label the best performance and use underlines to label the second. The NDCG and Recall metrics are abbreviated as `N' and `R', respectively.}
    \vspace{-0.2cm}
    \label{tab:overall}
    \resizebox{\linewidth}{!}{
    \begin{tabular}{*{17}{c}} 
        \hline \hline
        Dataset & Metric & BPRMF & NeuMF & GCMC & NGCF & GCCF & LightGCN & GraphCL & SGL & LightGCL & RocSE & CGI & SimGCL & RGCL & Improv. & \textit{p}-value\\
        \hline
        \multirow{6}{*}{ML-1M}  & R@10 & 0.1702 & 0.1553 & 0.1676 & 0.1763 & 0.1753 & 0.1774 & 0.1837 & 0.1828 & 0.1796 & 0.1786 & 0.1797 & \underline{0.1866} & \textbf{0.1934} & +3.91\% & 2.67e-4 \\ 
        & N@10 & 0.2485 & 0.2291 & 0.2480 & 0.2544 & 0.2624 & 0.2581 & 0.2617 & 0.2625 & 0.2591 & 0.2577 & 0.2613 & \underline{0.2657} & \textbf{0.2694} & +1.58\% & 7.52e-4 \\ \cline{2-17}
        & R@20 & 0.2582 & 0.2400 & 0.2526 & 0.2673 & 0.2611 & 0.2680 & 0.2749 & 0.2745 & 0.2722 & 0.2699 & 0.2703 & \underline{0.2798} & \textbf{0.2901 } & +3.69\% & 7.50e-4 \\ 
        & N@20 & 0.2576 & 0.2393 & 0.2551 & 0.2647 & 0.2677 & 0.2670 & 0.2721 & 0.2725 & 0.2693 & 0.2676 & 0.2699 & \underline{0.2758} & \textbf{0.2821 } & +2.29\% & 2.26e-3 \\ \cline{2-17}
        & R@50 & 0.4174 & 0.3952 & 0.4073 & 0.4297 & 0.4171 & 0.4310 & 0.4379 & 0.4381 & 0.4343 & 0.4333 & 0.4308 & \underline{0.4468} & \textbf{0.4581 } & +2.53\% & 4.42e-4 \\ 
        & N@50 & 0.3038 & 0.2848 & 0.2985 & 0.3121 & 0.3109 & 0.3137 & 0.3196 & 0.3202 & 0.3162 & 0.3149 & 0.3158 & \underline{0.3242} & \textbf{0.3321 } & +2.42\% & 4.08e-4 \\ \hline
        \multirow{6}{*}{Alibaba} & R@10 & 0.0682 & 0.0450 & 0.0503 & 0.0700 & 0.0707 & 0.0734 & 0.0741 & 0.0769 & 0.0747 & 0.0767 & 0.0740 & \underline{0.0791} & \textbf{0.0824} & +4.20\% & 1.69e-3 \\ 
        & N@10 & 0.0435 & 0.0284 & 0.0308 & 0.0446 & 0.0446 & 0.0461 & 0.0473 & 0.0486 & 0.0469 & 0.0485 & 0.0466 & \underline{0.0502} & \textbf{0.0528} & +5.00\% & 1.57e-4 \\ \cline{2-17}
        & R@20 & 0.1070 & 0.0718 & 0.0805 & 0.1101 & 0.1104 & 0.1138 & 0.1151 & 0.1187 & 0.1158 & 0.1166 & 0.1146 & \underline{0.1218} & \textbf{0.1267} & +4.00\% & 4.02e-4 \\ 
        & N@20 & 0.0553 & 0.0365 & 0.0399 & 0.0568 & 0.0567 & 0.0584 & 0.0598 & 0.0613 & 0.0594 & 0.0607 & 0.0589 & \underline{0.0632} & \textbf{0.0663} & +4.85\% & 1.54e-6 \\ \cline{2-17}
        & R@50 & 0.1875 & 0.1282 & 0.1454 & 0.1920 & 0.1931 & 0.1975 & 0.1944 & 0.2020 & 0.2010 & 0.1937 & 0.1967 & \underline{0.2059} & \textbf{0.2129} & +3.40\% & 4.63e-4 \\ 
        & N@50 & 0.0746 & 0.0501 & 0.0554 & 0.0764 & 0.0765 & 0.0784 & 0.0787 & 0.0812 & 0.0798 & 0.0792 & 0.0786 & \underline{0.0834} & \textbf{0.0869} & +4.29\% & 1.12e-4 \\ \hline
        \multirow{6}{*}{Kuaishou} & R@10 & 0.0565 & 0.0588 & 0.0645 & 0.0663 & 0.0787 & 0.0730 & 0.0738 & 0.0748 & 0.0775 & 0.0714 & 0.0726 & \underline{0.0788} & \textbf{0.0899} & +14.14\% & 5.05e-6 \\ 
        & N@10 & 0.0326 & 0.0351 & 0.0375 & 0.0370 & 0.0441 & 0.0413 & 0.0436 & 0.0450 & \underline{0.0461} & 0.0409 & 0.0417 & 0.0451 & \textbf{0.0498} & +8.00\% & 6.99e-4\\ \cline{2-17}
        & R@20 & 0.0992 & 0.1095 & 0.1193 & 0.1266 & 0.1327 & 0.1269 & 0.1225 & 0.1282 & \underline{0.1430} & 0.1242 & 0.1316 & 0.1325 & \textbf{0.1529} & +6.88\% & 4.03e-4 \\ 
        & N@20 & 0.0457 & 0.0504 & 0.0541 & 0.0551 & 0.0603 & 0.0573 & 0.0584 & 0.0609 & \underline{0.0660} & 0.0571 & 0.0596 & 0.0613 & \textbf{0.0687} & +4.09\% & 3.89e-3\\ \cline{2-17}
        & R@50 & 0.2027 & 0.2172 & 0.2203 & 0.2562 & 0.2477 & 0.2388 & 0.2366 & 0.2522 & \underline{0.2788} & 0.2489 & 0.2565 & 0.2503 & \textbf{0.2865} & +2.79\% & 8.94e-3 \\ 
        & N@50 & 0.0702 & 0.0760 & 0.0782 & 0.0857 & 0.0879 & 0.0840 & 0.0854 & 0.0902 & \underline{0.0980} & 0.0866 & 0.0891 & 0.0897 & \textbf{0.1005} & +2.54\% & 9.41e-3\\ 
        \hline
        \multirow{6}{*}{Gowalla} & R@10 & 0.1330 & 0.1205 & 0.1185 & 0.1296 & 0.1319 & 0.1419 & 0.1540 & 0.1470 & 0.1448 & 0.1461 & 0.1447 & \underline{0.1564} & \textbf{0.1606} & +2.66\% & 7.69e-4\\ 
        & N@10 & 0.1162 & 0.1038 & 0.1013 & 0.1136 & 0.1150 & 0.1257 & 0.1363 & 0.1305 & 0.1277 & 0.1271 & 0.1280 & \underline{0.1379} & \textbf{0.1419} & +2.89\% & 1.84e-3\\ \cline{2-17}
        & R@20 & 0.1894 & 0.1783 & 0.1749 & 0.1878 & 0.1924 & 0.2041 & 0.2178 & 0.2123 & 0.2085 & 0.2117 & 0.2059 & \underline{0.2245} & \textbf{0.2272} & +1.18\% & 1.83e-2 \\ 
        & N@20 & 0.1355 & 0.1238 & 0.1205 & 0.1333 & 0.1356 & 0.1470 & 0.1579 & 0.1527 & 0.1493 & 0.1495 & 0.1487 & \underline{0.1610} & \textbf{0.1646} & +2.22\% & 4.59e-3\\ \cline{2-17}
        & R@50 & 0.3003 & 0.2888 & 0.2832 & 0.3009 & 0.3057 & 0.3194 & 0.3335 & 0.3273 & 0.3240 & 0.3297 & 0.3205 & \underline{0.3460} & \textbf{0.3468} & +0.23\% & 1.31e-1\\ 
        & N@50 & 0.1682 & 0.1563 & 0.1524 & 0.1667 & 0.1691 & 0.1810 & 0.1922 & 0.1867 & 0.1835 & 0.1845 & 0.1826 & \underline{0.1969} & \textbf{0.2000} & +1.55\% & 1.58e-3\\ 
        \hline
        \multirow{6}{*}{Yelp} & R@10 & 0.0509 & 0.0407 & 0.0520 & 0.0506 & 0.0512 & 0.0612 & 0.0663 & 0.0681 & 0.0626 & 0.0656 & 0.0579 & \underline{0.0740} & \textbf{0.0753} & +1.75\% & 1.16e-2\\ 
        & N@10 & 0.0392 & 0.0309 & 0.0400 & 0.0390 & 0.0399 & 0.0479 & 0.0518 & 0.0532 & 0.0487 & 0.0512 & 0.0449 & \underline{0.0582} & \textbf{0.0591} & +1.58\% & 6.58e-3\\ \cline{2-17}
        & R@20 & 0.0844 & 0.0691 & 0.0867 & 0.0842 & 0.0851 & 0.1001 & 0.1067 & 0.1098 & 0.1021 & 0.1052 & 0.0940 & \underline{0.1182} & \textbf{0.1191} & +0.78\% & 1.52e-3 \\ 
        & N@20 & 0.0509 & 0.0408 & 0.0520 & 0.0507 & 0.0517 & 0.0614 & 0.0658 & 0.0677 & 0.0624 & 0.0650 & 0.0574 & \underline{0.0736} & \textbf{0.0744} & +1.09\% & 2.83e-3 \\ \cline{2-17}
        & R@50 & 0.1571 & 0.1339 & 0.1623 & 0.1570 & 0.1582 & 0.1814 & 0.1909 & 0.1950 & 0.1852 & 0.1871 & 0.1704 & \underline{0.2075} & \textbf{0.2108} & +1.58\% & 2.36e-3 \\ 
        & N@50 & 0.0720 & 0.0596 & 0.0740 & 0.0718 & 0.0730 & 0.0850 & 0.0903 & 0.0925 & 0.0865 & 0.0888 & 0.0796 & \underline{0.0995} & \textbf{0.1010} & +1.46\% & 2.03e-3 \\ 
        \hline\hline
    \end{tabular}
}
\end{table*}

\begin{figure}[t]
\centering
\includegraphics[width=\linewidth]{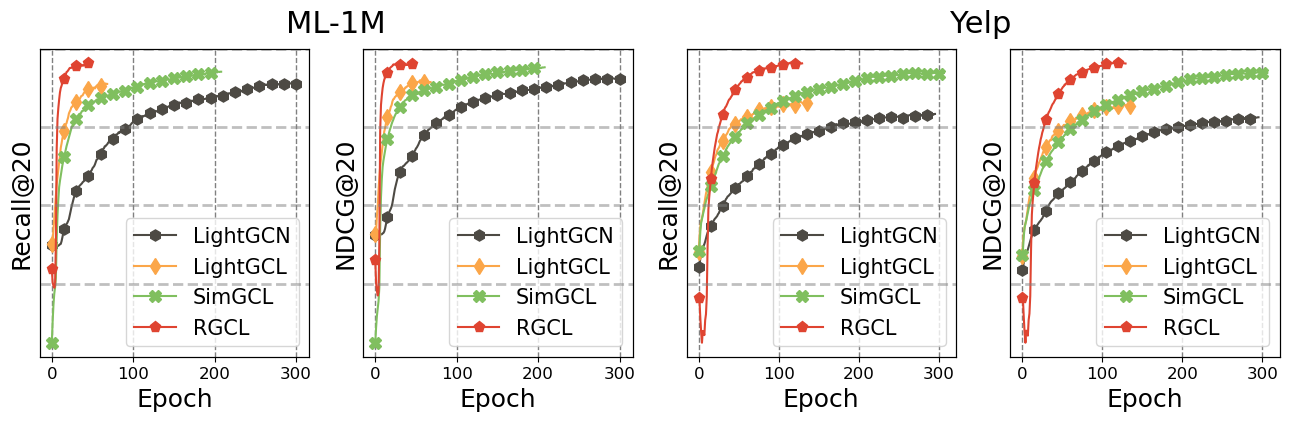} 
\caption{Model convergence analysis w.r.t training epochs on the ML-1M and Yelp datasets.} 
\label{fig:convergence}
\end{figure}

\section{EXPERIMENTS}
In this section, we conduct extensive experiments to validate the effectiveness of RGCL, and our goal is to answer the following research questions: 
\begin{itemize}[leftmargin=*]
\item \textbf{RQ1:} How does RGCL perform compared with state-of-the-art recommendation models?
\item \textbf{RQ2:} How do different designs of RGCL contribute to the final recommendation performance?
\item \textbf{RQ3:} How does RGCL perform against different data sparsity and item popularity?
\item \textbf{RQ4:} How do different hyper-parameters affect the recommendation performance of RGCL?
\end{itemize}

\subsection{Experimental Setup}
\textbf{Datasets.} We conduct extensive experiments on the following public recommendation datasets: MovieLens (ML)-1M~\cite{harper2015movielens}, Alibaba~\cite{chen2019pog}, Kuaishou~\cite{gao2022kuairand}, Gowalla~\cite{cho2011friendship}, and Yelp. For detailed introductions and preprocessing details of these datasets, please refer to Appendix~\ref{datasets}.

\noindent
\textbf{Baseline Models.} We compare RGCL with different state-of-the-art recommendation models, including traditional recommenders (BPR~\cite{rendle2012bpr} and NeuMF~\cite{he2017neural}), GNN-based recommenders (GCMC\cite{berg2017graph}, NGCF~\cite{wang2019neural}, GCCF~\cite{chen2020revisiting}, and LightGCN~\cite{he2020lightgcn}) and GCL-based recommenders (GraphCL~\cite{you2020graph}, SGL~\cite{wu2021self} , LightGCL~\cite{cai2023lightgcl}, CGI~\cite{wei2022contrastive}, RocSE~\cite{ye2023towards},  and SimGCL~\cite{yu2022graph}). The detailed introduction of all these baseline models are referred to Appendix~\ref{baselines}.

\noindent
\textbf{Evaluation Metrics.} To ensure the evaluation reliability, following standard practice~\cite{wei2022contrastive, wu2021self,yang2023generative}, we adopt the full-ranking strategy to mitigate the evaluation bias introduced by randomly negative sampling, which ranks all the items that are not interacted by the test user as candidate item pool. For evaluation metrics, we adopt the Normalized Discounted Cumulative Gain@$K$ (NDCG@$K$) and Recall@$K$, where $K\in\{10, 20, 50\}$.

For better reproducibility, more implementation details are provided in Appendix~\ref{implementation_details} and \url{https://tangjiakai.github.io/RGCL/}.

\subsection{Overall Performance (RQ1)}

The results of different methods on all datasets are shown in Table~\ref{tab:overall}. Based on the results, we have the following observations:
\begin{itemize}[noitemsep, nolistsep, label=$\bullet$, leftmargin=*]
    \item Compared to traditional baselines, such as BPRMF and NeuMF, all GNN-based models perform better on most datasets, which agrees with the previous work and confirms the effectiveness of GNNs~\cite{he2020lightgcn, wang2019neural}. Among all the GNN-based methods, LightGCN usually achieves the  excellent performance due to its simple yet effective linear convolution structure. Furthermore, most GCL-based recommenders outperform the GNN-based methods, indicating the desirable property of GCL for alleviating the bias introduced by high-degree nodes. 
    However, these GCL-based models fail to explicitly delineate the definitions of task-relevant semantic rationality and contrastive hardness, thus they achieve inferior balance between contrastive rationality and hardness when constructing augmentation views.
    \item By comparing our approach with all state-of-the-art baselines, it is clear to see that RGCL yields a consistent boost across all datasets. Besides, the most $p$-values that are much less than 0.01 also demonstrate the effectiveness of RGCL. 
    We attribute the marked enhancement in performance to the excellent balance between preserving semantic information and bolstering hardness of contrastive examples, which further prompts the ability upper bound of GCL-based recommenders. Besides, we increase the distance between sample points and decision boundary through enhanced adversarial examples, avoiding compromises in robustness caused by contrastive learning.
\end{itemize}

\noindent
\textbf{Training Efficiency.} Moreover, to verify the convergence performance of RGCL, we track the Recall@20 and NDCG@20 curves of different models \emph{w.r.t.} the training epochs in Figure~\ref{fig:convergence}. From the results, we can observe that RGCL converges significantly faster than SimGCL and LightGCN. Although LightGCL also achieves great convergence speed, its accuracy performance is worse than RGCL, as seen in \autoref{tab:overall}. One possible reason is that its static SVD contrastive view fails to keep pace with the evolving model capability during training, eventually limiting the improvement of representation quality. Different from these baselines, RGCL adopts the decision boundary-aware perturbation to guide on the example generation, which adaptively adjusts the augmentation strength to reduce the inconsistency between the representation quality and the contrastive hardness. As a result, RGCL shows both significantly greater efficiency and efficacy.

\subsection{Ablation Study (RQ2)}
To further validate the importance and contribution of each component in RGCL, we devise multiple simplified variants.
In specific, we compare the following four variants: (1) in \underline{w/o cons}, we drop the decision boundary-aware perturbation constraints on contrastive views. (2) In \underline{w/o rand}, we do not introduce random initialized perturbation (\emph{i.e.}, set $\mathbf{r}$ as all-one vector). (3) In \underline{w/o ac}, we drop the relation-aware view generator but only retain two random augmented views; (4) In \underline{w/o adv}, we drop the adversarial regularization term $\mathcal{L}_{ADV}$ in the final loss. The experiment is conducted based on the datasets of ML-1M and Yelp, while the observation and conclusion on the other datasets are similar and omitted.

We present the results in Table~\ref{tab:ablation}, where we can see: 
For \underline{w/o cons} variant, unconstrained perturbations result in a significant performance decrease, suggesting that a uniform perturbation cannot effectively preserve that semantic information due to different intrinsic robustness among instances. The \underline{w/o rand} variant performs much worse than RGCL, which demonstrates that introducing some variances for augmented views is necessary. 
Furthermore, our method gains improvement over \underline{w/o ac} variant, which reveals the importance of challenging positive pairs and hard negative pairs
However, only optimizing contrastive learning is still sub-optimal, which is evidenced by the lowered performance of \underline{w/o adv} variant as compared with RGCL. We speculate that over-optimizing contrastive learning for representation uniformity may potentially lead to a reduction in the distance between data points and the model's decision boundary, eventually deteriorating the robustness.
In summary, the above observations demonstrate that all the designs are crucial to the final performance improvement.

\begin{table*}[t]
    \renewcommand{\arraystretch}{1.1}
    \centering
    \caption{Ablation Study on ML-1M and Yelp datasets.}
    \label{tab:ablation}
    \begin{tabular}{c*{4}{c}|*{4}{c}} 
        \hline \hline
        \multirow{2}{*}{Model} & \multicolumn{4}{c|}{ML-1M} & \multicolumn{4}{c}{Yelp} \\
         & R@20 & N@20 & R@50 & N@50 & R@20 & N@20 & R@50 & N@50 \\
        \hline
         w/o cons & 0.2882 & 0.2798 & 0.4566 & 0.3302 & 0.1185 & 0.0733 & 0.2086 & 0.0995 \\
         w/o rand & 0.2838 & 0.2793 & 0.4470 & 0.3265 & 0.1183 & 0.0736 & 0.2080 & 0.0996 \\
         w/o ac & 0.2872 & 0.2813 & 0.4570 & 0.3315 & 0.1182 & 0.0737 & 0.2085 & 0.1000 \\
         w/o adv & 0.2832 & 0.2801 & 0.4470 & 0.3276 & 0.1180 & 0.0737 & 0.2083 & 0.1000  \\
         RGCL & \textbf{0.2901} & \textbf{0.2821} & \textbf{0.4581} & \textbf{0.3321} & \textbf{0.1191} & \textbf{0.0744} & \textbf{0.2108} & \textbf{0.1010} \\
        \hline \hline
    \end{tabular}
\end{table*}

\subsection{Robustness Evaluation (RQ3)}\label{RQ3}

\begin{figure}[t]
\centering
\includegraphics[width=\linewidth]{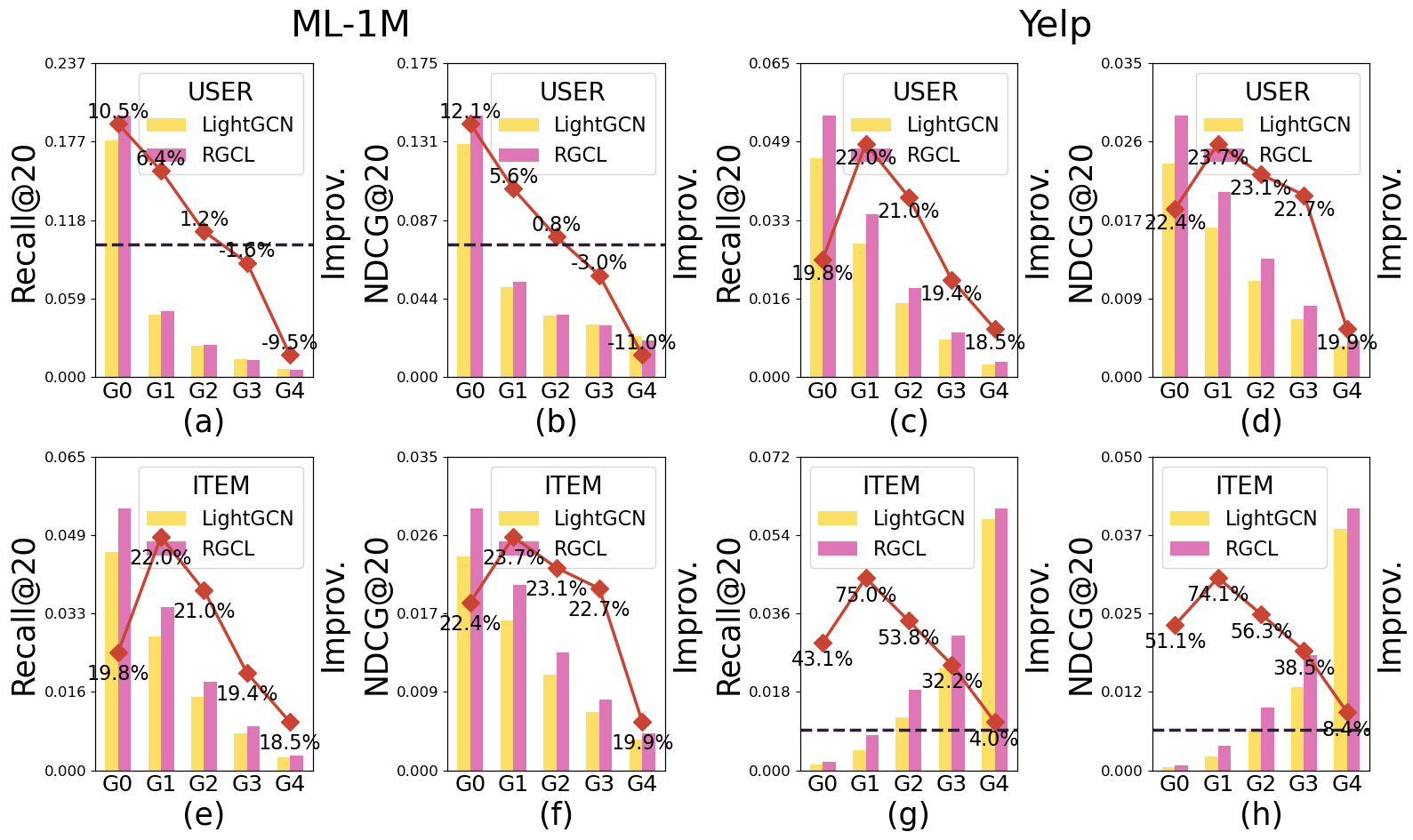} 
\caption{Recommendation performances at different level of data sparsity and item popularity. The black dashed line represents no performance improvement or decline.} 
\label{fig:sparsity}
\end{figure}

To validate the model robustness, we conduct experimental analysis based on different levels of user activity level and item popularity. For detailed user and item grouping approaches, please refer to Appendix~\ref{group}. The experimental results are presented in Figure~\ref{fig:sparsity}, where we can observe that in user (item) groups with sparse interactions, RGCL demonstrates more significant performance improvements. This implies that RGCL effectively capture interest preference of inactive users and characteristic of long-tailed items. Note that the performance trends on the item side for ML-1M and Yelp datasets are different. We speculate that one possible reason is that the proportion of long-tailed items in ML-1M is much higher than Yelp, which results in major contribution to the overall performance by low-degree item groups in ML-1M.

\subsection{Further Analysis of RGCL (RQ4)}
In this subsection, we further conduct more detailed experiments on the RGCL method to confirm its effectiveness. Due to space limitation, we only show the results on ML-1M and Yelp datasets while the similar conclusions can be derived from other datasets.

\begin{figure}[t]
\centering
\includegraphics[width=\linewidth]{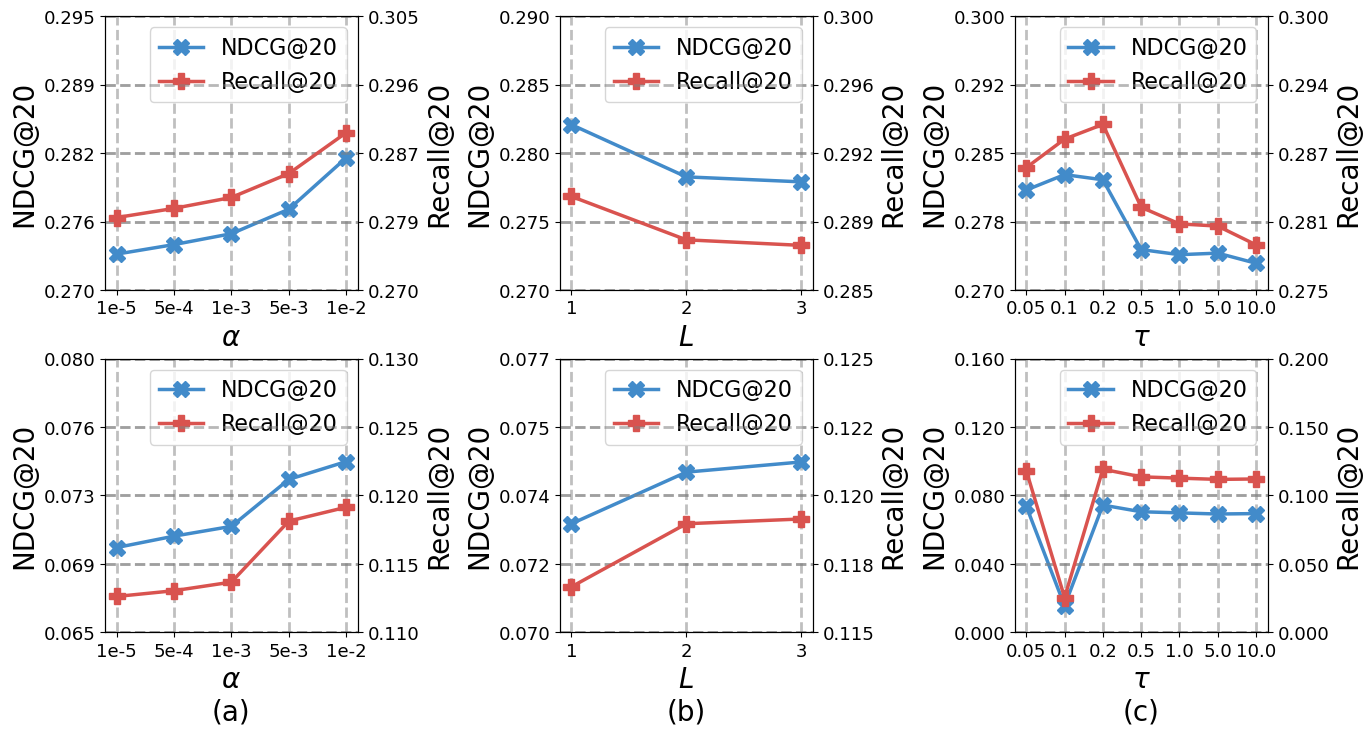} 
\caption{Hyper-parameter analysis w.r.t. $\alpha$, $L$, $\tau$. The top shows the experimental results on ML-1M and the bottom shows the results on Yelp.} 
\label{fig:hyper}
\end{figure}

\begin{figure}[t]
\centering
\includegraphics[width=\linewidth]{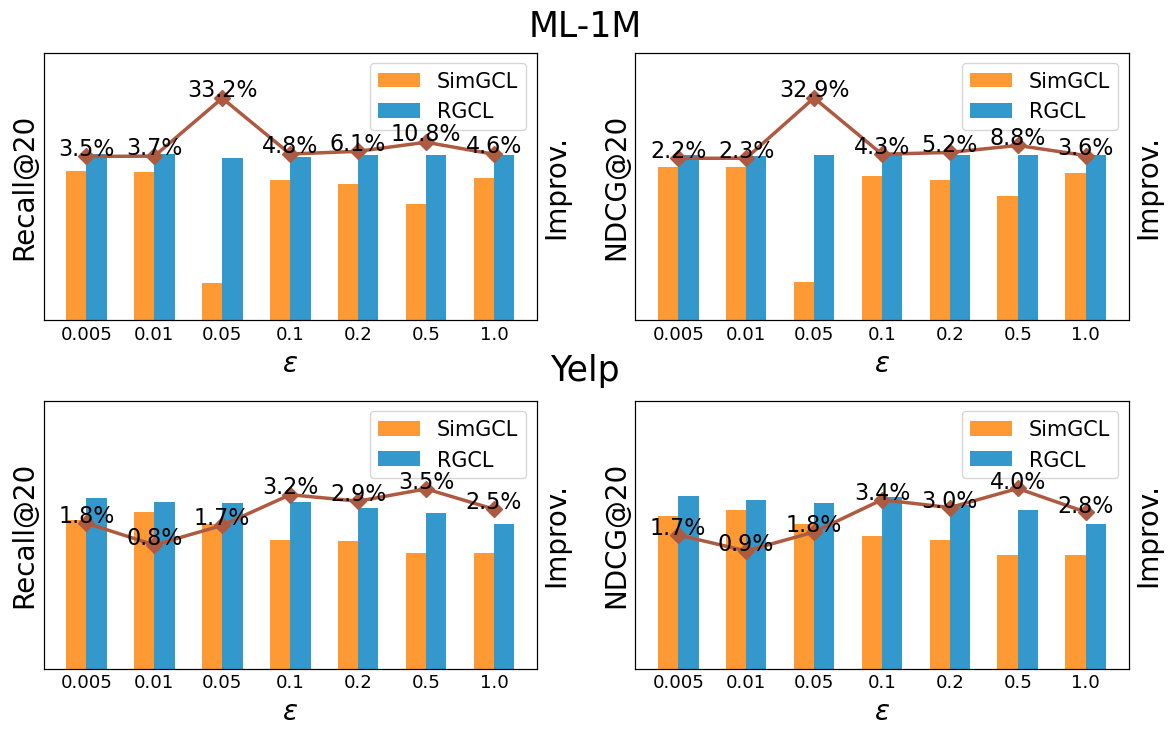} 
\caption{The model tolerance to hyper-parameter $\epsilon$ in terms of Recall@20 and NDCG@20 on ML-1M and Yelp datasets. The bars represent the accuracy metrics of different models (w.r.t. NDCG@20 and Recall@20), while the lines show the relative improvement of RGCL compared to SimGCL.} 
\label{fig:eps}
\end{figure}

\subsubsection{\textbf{Analysis of the model tolerance to hyper-parameter $\epsilon$}}

To validate the robustness of our method to perturbation hyper-parameter $\epsilon$, we conduct extensive experiments of performance comparison with SimGCL baseline with different values of $\epsilon$. Specifically, we set the search range as \{0.005,0.01,0.05,0.1,0.2,0.5,1.0\}. As shown in Figure~\ref{fig:eps}, we observe that SimGCL shows obvious performance fluctuations as $\epsilon$ changes. We speculate that the twofold reasons are the following: (1) different instances have different levels of intrinsic robustness. However, uniform and unconstrained perturbations may  potentially destroy the semantic structure for fragile instances, ultimately leading to erroneous contrastive views. (2) For instances with better intrinsic robustness, the hardness of contrastive examples is insufficient, hindering the full exploitation of contrastive learning. In contrast, our RGCL adopts decision boundary-aware perturbation constraints to guide the generation of both random and adversarial contrastive examples, leading to stable and superior performance. This demonstrates the insensitivity of RGCL to perturbation hyper-parameter $\epsilon$.

\subsubsection{\textbf{Impact of the coefficient $\alpha$}}
We change $\alpha$ to a set of predetermined representative values presented in Figure~\ref{fig:hyper}(a). We can see that the recommendation performance of RGCL gradually improves as $\alpha$ increases, which suggests that contrastive learning can facilitate the uniformity of node representation and learn high-quality features. Correlating with the results in Figure~\ref{fig:visual_item} and~\ref{fig:visual_user}, it also suggests that the personalized characteristic of low-degree users and items can be better captured by our algorithm.

\subsubsection{\textbf{Impact of the layer number $L$}}
To investigate the impact of the GNN layer number on model performance, we vary the hyper-parameter $L$ in the range $\{1,2,3\}$. From the Figure~\ref{fig:hyper}(b), We can observe that the performance trend of RGCL differs across different datasets. For example, for the ML-1M dataset, the over-smoothing issue occurs even with small value of $L$, while for the Yelp dataset, the model shows the significant performance improvement as graph layer number $L$ increases.

\subsubsection{\textbf{Impact of the temperature $\tau$}}
The temperature $\tau$ plays an important role in contrastive learning~\cite{wang2021understanding}. Figure~\ref{fig:hyper}(c) shows the impact of model performance \emph{w.r.t.} different $\tau$. We can see that the performance fluctuates severely as we use different $\tau$. Specifically, too large values of $\tau$ lead to poor performance, which is consistent with the previous work~\cite{wu2021self}. 
Conversely, too small temperature values also fail to achieve optimal model performance.
One possible reason is that too small $\tau$ enforces the model to concentrate few hardest examples that dominate the optimization process, which is detrimental to achieve the satisfactory generalization ability.
Therefore, a suitable temperature is essential to maximize the benefits from graph contrastive learning. 

\textbf{More Analysis.}
To comprehensively evaluate the superiority of RGCL, we conduct more extensive experiments in Appendix to answer the following research questions: 
\begin{itemize}[leftmargin=*]
\item \textbf{RQ5:} What is the effect of RGCL on improving the representation uniformity of users and items? (\emph{cf.} Appendix~\ref{visualize_of_rep})
\item \textbf{RQ6:} How does the RGCL framework perform when applied to other GNN backbones? (\emph{cf.} Appendix~\ref{generalization_exp})
\item \textbf{RQ7:} How does RGCL maintain the semantic information of contrastive examples? (\emph{cf.} Appendix~\ref{case_appendix})
\end{itemize}

\balance
\section{Related Work}

\textit{\textbf{Graph Neural Network in Recommendation.}} In recent years, the application of GNN models in recommender systems has achieved remarkable success~\cite{wang2019neural, he2020lightgcn, berg2017graph, chen2020revisiting}. 
For example, NGCF~\cite{wang2019neural} models the higher-order connectivity in user-item graph by explicitly injecting collaborative signals into the embedding process.
Compared with NGCF, LightGCN~\cite{he2020lightgcn} simplifies the design of GCN by removing redundant feature transformation and nonlinear activation function. 
However, GNN-based recommenders suffer from the sparsity of user-item interactions. Although external data sources (\emph{e.g.}, multi-behavior data and knowledge graphs) help mitigate the above issue, obtaining such data is often challenging and even unavailable due to expensive cost or privacy protection. In contrast, graph contrastive learning, as an popular self-supervised learning paradigm, effectively overcomes the challenge of data sparsity.

\noindent
\textit{\textbf{GCL-based Recommendation Models.}} Graph contrastive learning (GCL) bridges the advantages of GNN models with contrastive learning, effectively alleviating recommendation bias and simultaneously modeling high-order connectivity. Generally, GCL methods can be classified into hardness-driven models and rationality-driven methods.
Specifically, for hardness-driven methods, their key task is to construct diverse and challenging augmented views.
For example, GraphCL~\cite{you2020graph} and SGL~\cite{wu2021self} both devises multiple heuristic strategy to generate different contrastive views, such edge dropout and feature masking.
However, these methods are prone to losing important semantic features since the augmentation operations are indeed unrelated to the downstream task yet simply based on human-designed experiences. 
In contrast, rationality-driven GCL methods alleviate the above issue by introducing slight feature perturbations to maintain semantic consistency, such as SimGCL~\cite{yu2022graph} and RocSE~\cite{ye2023towards}.
However, these methods still suffer from potential issues, such as insufficient contrastive hardness and tedious trial-and-error of hyper-parameter, resulting in suboptimal performance and poor flexibility. Compared with these methods, our method achieves a better balance between rationality and hardness of contrastive examples via well-designed decision boundary-aware perturbations and adversarial-contrastive view-generator.

\section{Conclusion}
In this paper, we propose a novel graph contrastive learning framework, named RGCL, aiming to strike a better trade-off between rationality and hardness for the contrastive view-generator. 
Specifically, we propose a decision boundary-aware perturbation constraints and relation-aware adversarial-contrastive augmentation to generate contrastive examples. Besides, RGCL generates adversarial examples based on the adversarial perturbations to achieve margin maximization between data points and the decision boundary, further improving the model robustness. 
Finally, we design a joint optimization objective to optimize model parameters.

\section*{Acknowledgments}
This work is supported in part by National Key R\&D Program of China (2023YFF0905402), National Natural Science Foundation of China (No. 62102420), Beijing Outstanding Young Scientist Program NO. BJJWZYJH012019100020098, Intelligent Social Governance Platform, Major Innovation \& Planning Interdisciplinary Platform for the ``DoubleFirst Class'' Initiative, Renmin University of China, Public Computing Cloud, Renmin University of China, fund for building world-class universities (disciplines) of Renmin University of China, Intelligent Social Governance Platform. The work is sponsored by KuaiShou Technology Programs (No. 2022020091).

\clearpage

\bibliographystyle{ACM-Reference-Format}
\bibliography{main}

\clearpage

\appendix
\section*{Content of Appendix}
\begin{enumerate}[label=\bfseries\Alph*,leftmargin=*]
    \item \hyperref[time_complexity]{Analysis of Training Time Complexity}
    \item \hyperref[theory]{Further Robustness ANALYSIS}
    \item \hyperref[experiment_details]{Experimental Details}
    \begin{enumerate}[label=\bfseries C.\arabic*,leftmargin=*]
        \item \hyperref[datasets]{Recommendation Datasets}
        \item \hyperref[baselines]{Baselines}
        \item \hyperref[implementation_details]{Implementation Details}
        \item \hyperref[group]{Details on User and Item Grouping}
        \item \hyperref[alg_process]{Learning Algorithm of RGCL}
    \end{enumerate}
    \item \hyperref[more_exp]{More Experiment Analysis}
    \begin{enumerate}[label=\bfseries D.\arabic*,leftmargin=*]
        \item \hyperref[visualize_of_rep]{Visualization of Representation (RQ5)}
        \item \hyperref[generalization_exp]{Generalization Evaluation (RQ6)}
        \item \hyperref[case_study]{Case Study (RQ7)}
    \end{enumerate}
\end{enumerate}

\section{Analysis of Training Time Complexity}\label{time_complexity}
The extra training time complexity of RGCL comes from the loss terms of contrastive and adversarial components. Suppose the number of nodes and edges are $|\mathcal{V}|$ and $|\mathcal{E}|$, respectively. Let $B$ denote the batch size, $d$ denote the embedding dimension, L denote the total layer number. We analyze the time complexity of each component as follows:
\begin{itemize}[leftmargin=*]
    \item \textbf{Original loss}. The time complexity of the original LightGCN model comes from adjacent matrix construction, graph convolution computation and BPR calculation.
    Their time complexities are $O(|\mathcal{E}|)$, $O(L|\mathcal{E}|d)$ and $O(Bd)$ respectively. Therefore, the total time complexity is $O((L|\mathcal{E}|+B)d)$.
    \item \textbf{Contrastive loss}. To begin with, solving for the perturbation constraints in contrastive learning needs one pass of forward and backward propagation, where the time complexity is $O(L|\mathcal{E}|d)$. Then, constructing two random-augmented views requires two pass of forward propagation. As for adversarial-contrastive view, it also needs extra one pass of forward and backward propagation, where the time complexity of the contrastive loss paradigm is $O(B^2d)$. Therefore, the total time complexity of the contrastive learning component is $O((L|\mathcal{E}|+B^2)d)$.
    \item \textbf{Adversarial loss}. The adversarial perturbations for generating adversarial examples has already been accounted in the contrastive loss part. Thus, in this part, we simply consider the time complexity of forward propagation and BPR loss, which are $O(L|\mathcal{E}|d)$ and $O(Bd)$, respectively. Therefore, the total time complexity of the adversarial loss is $O((L|\mathcal{E}|+B)d)$.
\end{itemize}

In summary, the total time complexity of the proposed RGCL is $O((L|\mathcal{E}|+B^2)d)$, which maintains the same order of time complexity as other graph contrastive learning algorithms~\cite{yu2022graph, ye2023towards}. However, the experimental results in Figure~\ref{fig:convergence} demonstrates that our algorithm has better converge and accuracy performance.

\section{Further Robustness ANALYSIS}\label{theory}
Inspired by previous work~\cite{neyshabur2017exploring,xia2022simgrace}, we provide the robustness analysis from the perspective of connections between sharpness of loss landscape and PAC-Bayes theory. Generally, smoother feature space can avoid large feature variations caused by input perturbations~\cite{wen2020towards}. Meanwhile, from the perspective of model optimization, flatter loss landscape can bring better model robustness. Specifically, 
assuming that the prior distribution $\mathcal{Q}$ over the model parameters, with probability at least $1-\xi$ over the draw of the training data, the expected error of $\mathcal{L}_{BPR}$ can be bounded as follows:
\begin{equation}
    \mathbb{E}_{\mathbf{\Delta}}\left[\widetilde{\mathcal{L}}_{BPR} \right] \leq \mathbb{E}_{\mathbf{\Delta}}\left[\mathcal{L}_{BPR}\right] + 4\sqrt{\frac{\mathrm{KL}(\bm{\theta} + \bm{\xi} \| \mathcal{Q}) + \ln \frac{2m}{\xi}}{m}},
\end{equation}
where $\widetilde{\mathcal{L}}_{BPR}$ represents the expected error, $m$ is the size of training data, $\mathbf{\Delta}$ denotes the perturbation of model parameter.
Then, we rewrite the above bound as follows:
\begin{equation}
    \begin{aligned}
        \mathbb{E}_{\mathbf{\Delta}}\left[\widetilde{\mathcal{L}}_{BPR} \right] \leq 
        \mathbb{E}\left[\mathcal{L}_{BPR}\right] 
        + \underbrace{\mathbb{E}_{\mathbf{\Delta}}\left[\mathcal{L}_{BPR}\right]
        - \mathbb{E}\left[\mathcal{L}_{BPR}\right]}_{\text{Expected sharpness}} \\
        + 4\sqrt{\frac{\mathrm{KL}(\bm{\theta} + \bm{\Delta} \| \mathcal{Q}) + \ln \frac{2m}{\xi}}{m}},
    \end{aligned}
\end{equation}
where expected sharpness $\mathbb{E}_{\mathbf{\Delta}}\left[\mathcal{L}_{BPR}\right]-\mathbb{E}\left[\mathcal{L}_{BPR}\right]$ demonstrates that our method aims to reduce the sensitivity to model parameter variations and increase the smoothness of the feature space. Therefore, the proposed perturbation-based augmentation examples can achieve more robust and well-generalized model performance.

\begin{table}[t]
    \renewcommand{\arraystretch}{1.1}
    \centering
    \caption{Statistics of the datasets.}
    \label{tab:dataset}
    \begin{tabular}{*{5}{c}} 
        \hline \hline
        \textbf{Dataset} & \textbf{\#Users} & \textbf{\#Items} & \textbf{\#Interactions} & \textbf{Sparsity} \\
        \hline
        ML-1M & 6,038 & 3,489 & 820,336 & 96.1059\% \\
        Alibaba & 12,265 & 6,145 & 193,120 & 99.7437\% \\
        Kuaishou & 2,457 & 1,042 & 35,795 & 98.6019\% \\
        Gowalla & 13,149 & 14,009 & 535,650 & 99.7092\% \\
        Yelp & 42,324 & 28,748 & 1,611,965 & 99.8675\% \\
        \hline \hline
    \end{tabular}
\end{table}

\section{EXPERIMENT DETAILS}\label{experiment_details}
\subsection{Recommendation Datasets}\label{datasets}
We conduct extensive experiments on the following five publicly available recommendation datasets in this paper:
(1) \textbf{MovieLens (ML)-1M\footnote{https://grouplens.org/datasets/movielens/}} is a widely adopted movie recommendation dataset, containing the one million movie ratings provided by users, ranging from 1 to 5 stars. (2) \textbf{Alibaba\footnote{https://github.com/wenyuer/POG}} is a fashion-related dataset and provides user behaviors related to both outfits and fashion items. (3) \textbf{Kuaishou\footnote{https://kuairand.com/}} contains user interactions on exposed short videos, collected from the video-sharing mobile App. (4) \textbf{Gowalla\footnote{https://snap.stanford.edu/data/loc-gowalla.html}} is a checking-in dataset for item recommendation, collected from a location-based social networking website. (5) \textbf{Yelp\footnote{https://www.kaggle.com/datasets/yelp-dataset/yelp-dataset/versions/2?resource=download}} is a widely-used business recommendation dataset collected from yelp website, where the business venues of users are viewed as the items.

To transform the explicit user ratings into implicit interaction behavior, the interactions with ratings above three are viewed as the positive example for rating-based datasets (\emph{i.e.}, ML-1M and Yelp). For Yelp and Gowalla datasets, we filter users and items that have less than fifteen interaction number to ensure the data quality. For all datasets, we randomly divide the data into training set, validation set and testing set using a ratio of 8:1:1. For negative samples used in BPR objective, we uniformly sample one negative item for each positive interaction. The overall experiments are repeated five times with different initialized seeds for significance test of model performance. The statistics of the five recommendation datasets are shown in Table~\ref{tab:dataset}.

\subsection{Baselines}\label{baselines}

\underline{\textbf{Traditional Recommenders:}}
\begin{itemize}[leftmargin=*]
    \item \textbf{BPRMF}~\cite{rendle2012bpr} is a well known matrix factorization model by optimizing BPR loss function.
    \item \textbf{NeuMF}~\cite{he2017neural} is a deep recommendation model, which aims to capture the non-linear correlations between users and items.
\end{itemize}

\underline{\textbf{GNN-based Recommenders:}}
\begin{itemize}[leftmargin=*]
    \item \textbf{GCMC}~\cite{berg2017graph} is a graph auto-encoder framework to learn complex patterns and dependencies within the user-item interaction graph by differentiable message passing.
    \item \textbf{NGCF}~\cite{wang2019neural} is a collaborative filtering model that integrates interactions of user-item bipartite into the embedding process for modeling high-order connectivity. 
    \item \textbf{GCCF}~\cite{chen2020revisiting} is a linear graph recommendation model, which alleviates the over smoothing problem by removing non-linearity and introducing the residual network structure.
    \item \textbf{LightGCN}~\cite{he2020lightgcn} is a graph-based recommender model, which enhances the collaborative filtering information by abandoning the feature transformation and nonlinear activation.
\end{itemize}

\underline{\textbf{GCL-based Recommenders:}}
\begin{itemize}[leftmargin=*]
    \item \textbf{GraphCL}~\cite{you2020graph} is a graph contrastive learning framework, which designs various types of graph augmentations to incorporate transformation randomness (\emph{e.g.}, attribute masking). 
    \item \textbf{SGL}~\cite{wu2021self} is a self-supervised learning method based on user-item bipartite interaction graph, which devises three augmentation strategies, \emph{aka.}, node dropout, edge dropout and random walk.
    \item \textbf{LightGCL}~\cite{cai2023lightgcl} is a simple graph contrastive paradigm that utilizes the SVD for contrastive augmentation to integrate the global collaborative relation without structural refinement.
    \item \textbf{RocSE}~\cite{ye2023towards} is a robust graph collaborative filtering model, which adds in-distribution perturbation to construct a contrastive view-generator, which mimicking the behaviors of adversarial attacks.
    \item \textbf{CGI}~\cite{wei2022contrastive} is a graph contrastive model by designing learnable graph augmentation to adaptively learn whether to drop an edge or node and leveraging the information bottleneck technique to guide contrastive learning process.
    \item \textbf{SimGCL}~\cite{yu2022graph} is a GCL-based recommendation model, which discards the sophisticated graph augmentation and adopts to add uniform noises to the embedding space as contrastive views.
\end{itemize}

\subsection{Implementation Details}\label{implementation_details}
We implement our RGCL with PyTorch~\cite{paszke2019pytorch} framework. For fair comparison, all models are initialized with the Xavier method~\cite{glorot2010understanding} and optimized by the Adam optimizer~{\cite{kingma2014adam}. All hyper-parameters of baseline models are searched following suggestions from the original papers. The batch size and embedding dimension are fixed to 4,096 and 64, respectively. The learning rate is searched from $\{0.0005, 0.001, 0.005, 0.01, 0.05\}$. The layer number of graph neural network is searched from $\{1,2,3\}$. We set $\mu=0.1$ in Equation~(\ref{final_loss}). The loss weight $\alpha$ is tuned from $\{1e-5,5e-5,\dots,1e-2\}$. The initial hyper-parameter used for perturbation magnitude is chosen from $\{0.005,0.01,\cdots,1.0\}$. The search range of temperature coefficient $\tau$ is $\{0.05,0.1,0.2,0.5,1.0,5.0,10.0\}$. Early stopping is utilized as the convergence criterion. Specifically, we evaluate the performance on the validation dataset for each epoch, and stop the training process once there is no accuracy improvement for 10 consecutive epochs.

\begin{table*}[t]
    \renewcommand{\arraystretch}{1.1}
    \centering
    \caption{Generalization evaluation on different GNN-based backbones.}
    \label{tab:generalization}
    \resizebox{1\textwidth}{!}{
    \begin{tabular}{c*{6}{c}|*{6}{c}} 
        \hline \hline
        \multirow{2}{*}{Model} & \multicolumn{6}{c|}{ML-1M} & \multicolumn{6}{c}{Yelp} \\
         & R@10 & N@10 & R@20 & N@20 & R@50 & N@50 & R@10 & N@10 & R@20 & N@20 & R@50 & N@50 \\
        \hline
        GCMC & 0.1676 & 0.2480 & 0.2526 & 0.2551 & 0.4073 & 0.2985 & 0.0520 & 0.0400 & 0.0867 & 0.0520 & 0.1623 & 0.0740 \\
        GCMC + RGCL & \textbf{0.1807} & \textbf{0.2608} & \textbf{0.2714} & \textbf{0.2707} & \textbf{0.4351} & \textbf{0.3176} & \textbf{0.0596} & \textbf{0.0463} & \textbf{0.0980} & \textbf{0.0596} & \textbf{0.1802} & \textbf{0.0835} \\
        Improv. & +7.86\% & +5.15\% & +7.42\% & +6.11\% & +6.82\% & +6.42\% & +14.60\% & +15.65\% & +13.02\% & +14.44\% & +11.03\% & +12.82\% \\
        \hline
        NGCF & 0.1763 & 0.2544 & 0.2673 & 0.2647 & 0.4297 & 0.3121 & 0.0506 & 0.0390 & 0.0842 & 0.0507 & 0.1570 & 0.0718 \\
        NGCF + RGCL & \textbf{0.1813} & \textbf{0.2565} & \textbf{0.2744} & \textbf{0.2683} & \textbf{0.4378} & \textbf{0.3165} & \textbf{0.0530} & \textbf{0.0405} & \textbf{0.0878} & \textbf{0.0526} & \textbf{0.1662} & \textbf{0.0752} \\
        Improv. & +2.83\% & +0.81\% & +2.67\% & +1.36\% & +1.89\% & +1.41\% & +4.87\% & +3.86\% & +4.23\% & +3.71\% & +5.82\% & +4.72\%  \\
        \hline
        GCCF & 0.1753 & 0.2624 & 0.2611 & 0.2677 & 0.4171 & 0.3109 & 0.0512 & 0.0399 & 0.0851 & 0.0517 & 0.1582 & 0.0730  \\
        GCCF + RGCL & \textbf{0.1838} & \textbf{0.2679} & \textbf{0.2722} & \textbf{0.2747} & \textbf{0.4315} & \textbf{0.3195} & \textbf{0.0575} & \textbf{0.0451} & \textbf{0.0937} & \textbf{0.0576} & \textbf{0.1701} & \textbf{0.0798} \\
        Improv. & +4.84\% & +2.09\% & +4.25\% & +2.61\% & +3.47\% & +2.76\% & +12.34\% & +12.98\% & +10.15\% & +11.49\% & +7.54\% & +9.32\% \\
        \hline
        LightGCN & 0.1774 & 0.2581 & 0.2680 & 0.2670 & 0.4310 & 0.3137 & 0.0612 & 0.0479 & 0.1001 & 0.0614 & 0.1814 & 0.0850 \\
        LightGCN + RGCL & \textbf{0.1934} & \textbf{0.2694} & \textbf{0.2901} & \textbf{0.2821} & \textbf{0.4581} & \textbf{0.3321} & \textbf{0.0753} & \textbf{0.0591} & \textbf{0.1191} & \textbf{0.0744} & \textbf{0.2108} & \textbf{0.1010} \\
        Improv. & +9.02\% & +4.39\% & +8.26\% & +5.65\% & +6.29\% & +5.86\% & +22.89\% & +23.39\% & +19.05\% & +21.19\% & +16.20\% & +18.84\% \\
        \hline \hline
    \end{tabular}}
\end{table*}

\renewcommand{\algorithmicensure}{\textbf{Output:}}
\renewcommand{\algorithmicrequire}{\textbf{Input:}}
\algnewcommand{\algorithmicparameter}{\textbf{Parameter:}}
\algnewcommand{\Parameter}{\item[\algorithmicparameter]}
\begin{algorithm}
\caption{Learning Algorithm of RGCL}
\label{alg}
\begin{algorithmic}[1]
\Require User-item bipartite graph $\mathcal{G} = \{\mathcal{V}, \mathbf{A}\}$, adversarial loss weight $\mu$, contrastive loss weight $\alpha$, initialized perturbation magnitude $\epsilon$, temperature coefficient $\tau$, layer number $K$, batch size $B$, learning rate $lr$;
\Parameter Learnable parameters $\bm{\theta}=\mathbf{E}$,
\Ensure RGCL Model;
\While{Model Not Convergence}
    \Statex \hspace*{1.3em} // Calculate the decision boundary-aware perturbation
    \State Calculate the perturbation $\mathbf{\Delta}_u^{(k)},\mathbf{\Delta}_i^{(k)}$ using Eq.~(\ref{max_pert_vec});
    \Statex \hspace*{1.3em} // Calculate the contrastive loss
    \State Generate perturbation-constrained random views $\mathbf{z}_u^{\prime}$,  $\mathbf{z}_u^{\prime\prime}$, $\mathbf{z}_i^{\prime}$,  $\mathbf{z}_i^{\prime\prime}$ using Eq.~(\ref{max_pert_vec}) and (\ref{random_aug}); 
    \State Generate relation-aware adversarial-contrastive views $\mathbf{z}_u^{ac}, \mathbf{z}_i^{ac}$ using Eq.~(\ref{eta_fgsm}) and (\ref{ac_view});
    \State Calculate multi-view contrastive loss $\mathcal{L}_{CL}$ using Eq.~(\ref{multi_view});
    \Statex \hspace*{1.3em} // Calculate the adversarial loss
    \State Generate adversarial examples $\mathbf{z}_u^{adv}$ and $\mathbf{z}_i^{adv}$ using Eq.~(\ref{adv_example});
    \State Calculate adversarial loss $\mathcal{L}_{ADV}$ using Eq.~(\ref{adv_loss});
    \Statex \hspace*{1.3em} // Calculate the BPR loss;
    \State Calculate the BPR loss $\mathcal{L}_{BPR}$ using Eq.~(\ref{bpr});
    \Statex \hspace*{1.3em} // Model optimization
    \State Calculate total loss $\mathcal{L}$ using Eq.~(\ref{final_loss});
    \State Update model parameter $\bm{\theta}$ using SGD;
\EndWhile
\State \Return $\bm{\theta}$;
\end{algorithmic}
\end{algorithm}

\subsection{Details on User and Item Grouping}\label{group}
In the following, we provide the specific details of partitioning the user and item groups in Experiment~\ref{RQ3}:
\begin{itemize}[leftmargin=*]
\item \textbf{USER}: we split all users into five groups based on the number of user interaction while keeping the total number of each user group the same, which are denoted as $[G_0,G_1,G_2,G_3,G_4]$ in ascending order of interaction count.
\item \textbf{ITEM}: we group all items based on their popularity into five groups and similarly, we keep the total number of each item group the same. Specifically, we adopt the decomposed Recall and NDCG metrics defined as follows:
\begin{gather*}
    \text{Recall}(\text{G}_i)=\frac{1}{M}\sum_{u\in \mathcal{U}}\frac{|\hat{l}_u \cap l_u^{\text{G}_i}|}{|\hat{l}_u|}, \\
    \text{NDCG}(\text{G}_i) = \frac{1}{M} \sum_{u\in\mathcal{U}} \frac{\sum_{j=1}^{|\hat{l}_u|}\mathbb{I}(\hat{l}_u(j)\in l_u^{\text{G}_i})(\log_2(j+1))^{-1}}{\sum_{t=1}^{|\hat{l}_u|} (\log_2(t+1))^{-1}},
\end{gather*}
where  $\hat{l}_u$ and $l_u$ represent the predicted and real Top-N recommendation list of user $u$, respectively, and $\mathbb{I}(\cdot)$ is the indication function. We use $l_u^{\text{G}_i}$ to denote the item recommendation list within the group $\text{G}_i$. Here, we set $|\hat{l}_u|=\min(|l_u|,K)$.
\end{itemize}

\begin{figure}[t]
\centering
\includegraphics[width=\linewidth]{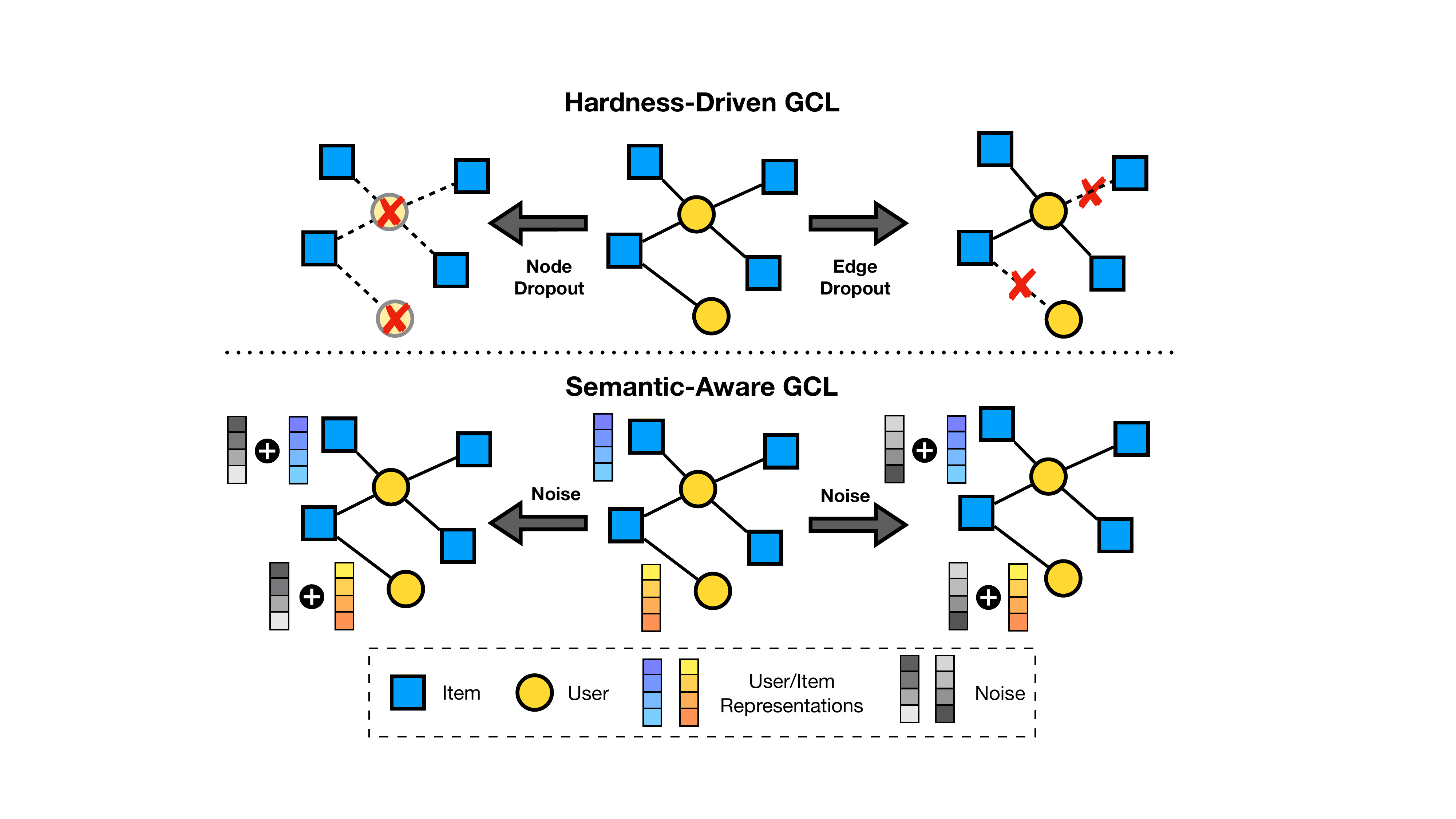} 
\caption{Visualization of item representation and degree on ML-1M and Yelp datasets. Darker colors indicate more points falling within the region.} 
\label{fig:visual_item}
\end{figure}

\subsection{Learning Algorithm of RGCL}\label{alg_process}
The overall learning algorithm of the proposed RGCL framework is summarized in Algorithm~\ref{alg}.

\section{More Experimental Analysis}\label{more_exp}

\subsection{Visualization of Representation (RQ4)}\label{visualize_of_rep}
To better understand how RGCL promotes the uniformity of representations for preserving personalized node information, we visualize the learned item embeddings and user embeddings in Figure~\ref{fig:visual_item} and Figure~\ref{fig:visual_user}, respectively.
Specifically, we firstly map the learned node representations to 2-dimensional normalized vectors using t-SNE~\cite{van2008visualizing}. Then, we use Kernel Density Estimation (KDE)~\cite{botev2010kernel} to visualize the distribution of transformed feature representations. Moreover, for a clearer demonstration, we also visualize the density estimations of their angles, where angles are calculated using the function: $arctan2(y,x)$ for each instance $(x,y)$.
We can observe our RGCL shows a better uniform distribution on both users and items. This shows that RGCL can effectively learn high-quality representations by avoiding the bias caused by the dominance of advantaged users and items. Besides, correlating with the results in Table~\ref{tab:overall}, RGCL achieves a win-win breakthrough in representation uniformization and performance improvement compared other baselines, suggesting the superiority of our designs.

\begin{figure}[t]
\centering
\includegraphics[width=\linewidth]{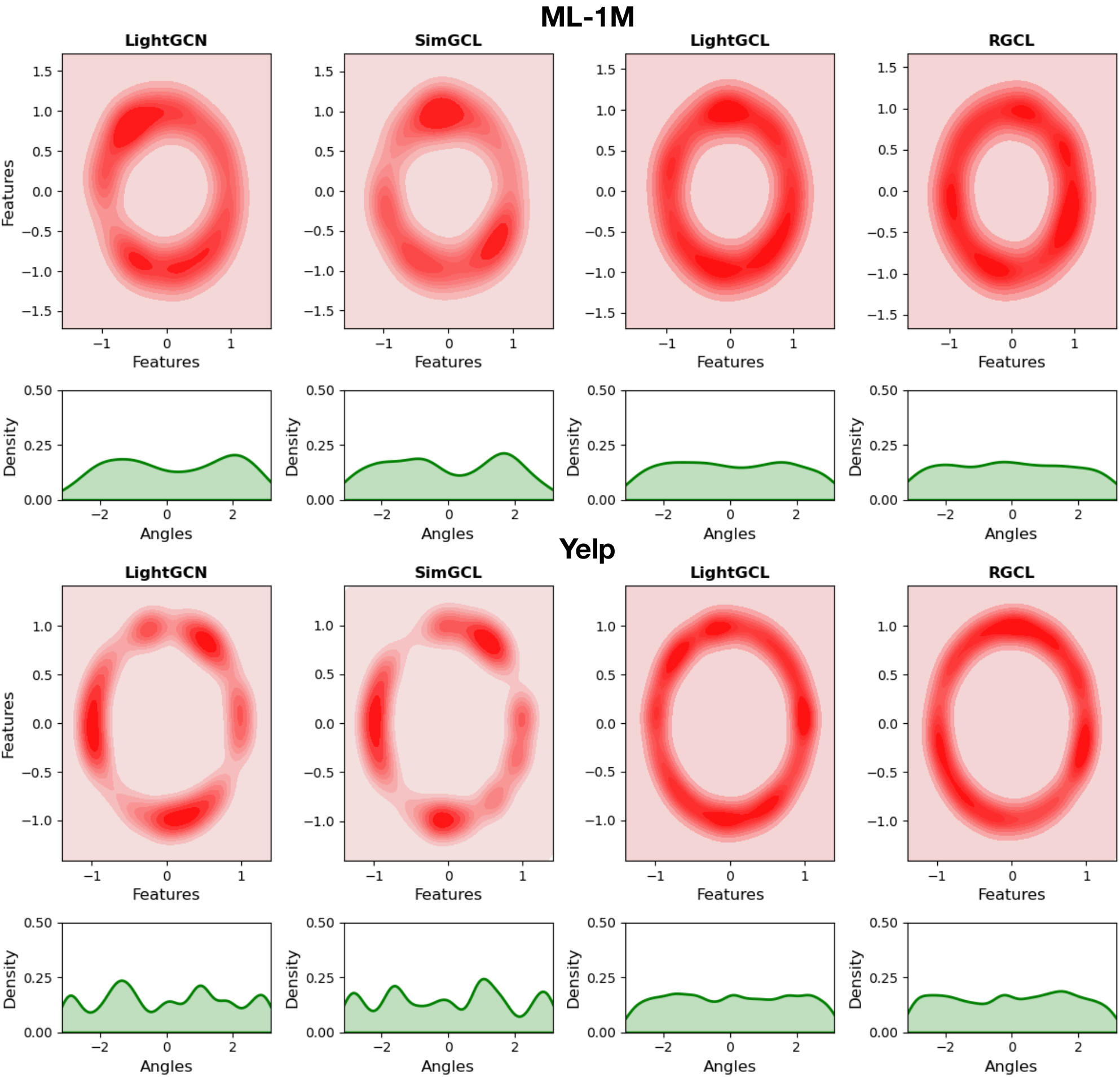} 
\caption{Visualization of user representation and degree on ML-1M and Yelp datasets. Darker colors indicate more points falling within the region.} 
\label{fig:visual_user}
\end{figure}

\subsection{Generalization Evaluation (RQ5)}\label{generalization_exp}
To verify the generalization of our proposed model-agnostic framework, we employ RGCL framework on three other commonly used GNN-based backbones, \emph{i.e.}, GCMC~\cite{berg2017graph}, NGCF~\cite{wang2019neural} and GCCF~\cite{chen2020revisiting}. 
We summarize the experimental results in Table~\ref{tab:generalization}. 
From the table, we can see that RGCL generalizes well across different GNN-based backbones, further demonstrating the effectiveness and flexibility of our method. Additionally, the improvement based on the NGCF backbone is not significant, which we attribute to the redundant weight parameters and unnecessary nonlinear feature transformations of NGCF model, thus posing challenges to the model learning.

\subsection{Case Study (RQ7)}\label{case_appendix}
In this section, we present a case study to intuitively show the effectiveness of our model to preserve the important semantic information of recommendation task. From the Figure~\ref{fig:case_ml}, we can observe that user \#315 prefers horror, action, and science fiction movies while showing less interest in comedy movies. Comparing the SimGCL and RGCL methods, although both original ranking results attain the correct ordering preferences for positive items and negative items, the introduction of noise perturbation for SimGCL baseline leads to a reversal in the predicted scores for movies \#757 (liked movie) and movie \#642 (disliked movie). It indicates that SimGCL baseline cannot reasonably control perturbations to preserve task-relevant information, resulting in irrational contrastive samples. In contrast, our proposed RGCL generates rational contrastive pairs and thus effectively improves model robustness and recommendation performance.

\begin{figure}[t]
\centering
\includegraphics[width=\linewidth]{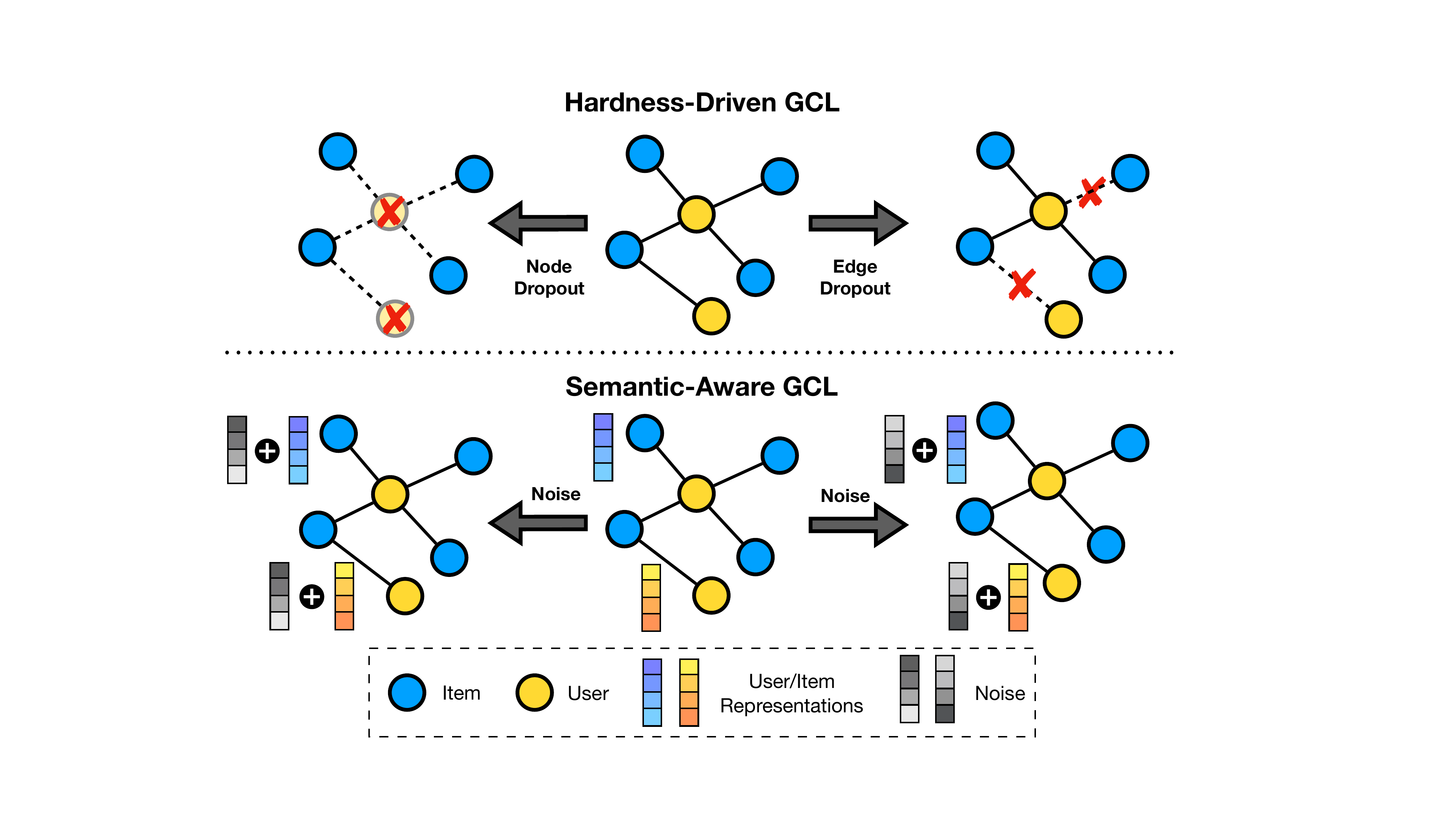} 
\caption{Case study on ML-1M dataset. The "Score (Origin.)" and "Score (Pert)" indicate predicted scores based on the original and contrastive augmented user and item embeddings, respectively. Best viewed in color.} 
\label{fig:case_ml}
\end{figure}

\end{document}